\documentclass[11pt]{article}
\usepackage{tikz}
\usepackage{ucs}
\usepackage{amssymb}
\usepackage{amsthm}
\usepackage{amsmath}
\usepackage{latexsym}
\usepackage[cp1251]{inputenc}
\usepackage[english]{babel}
\usepackage{graphicx}
\usepackage{wrapfig}
\usepackage{caption}
\usepackage{subcaption}
\usepackage{txfonts}
\usepackage{lmodern}
\usepackage{mathrsfs}
\usepackage{algorithm}
\usepackage{multicol}
\usepackage{enumerate}
\usepackage{titlesec}
\usepackage{paralist}
\usepackage{mathtools}
\titleformat{\paragraph}
{\normalfont\normalsize\bfseries}{\theparagraph}{1em}{}
\titlespacing*{\paragraph}
{0pt}{3.25ex plus 1ex minus .2ex}{1.5ex plus .2ex}

\newcommand*{\myproofname}{Proof}

\usepackage{wasysym}

\usepackage{fullpage}

\def\qed{\hfill\ifhmode\unskip\nobreak\fi\qquad\ifmmode\Box\else\hfill$\Box$\fi}

\usepackage{amsthm, amssymb, amsmath}
\usepackage{verbatim, float}
\usepackage{mathrsfs}
\usepackage{graphics}
\usepackage[usenames,dvipsnames]{pstricks}
\usepackage{epsfig}
\usepackage{pst-grad} 
\usepackage{pst-plot} 
\usepackage{cases}
\usepackage{algorithmic}
\usepackage{subcaption}
\usepackage{tikz,pgfplots}
\usepackage{chngcntr}

\newcommand*\xor{\oplus}

\newlength\figureheight 
\newlength\figurewidth 
\usepackage[top=1.8cm, bottom=1.8cm, left=1.8cm, right=1.8cm]{geometry}

\newtheorem{theorem}{Theorem}
\newtheorem{corollary}[theorem]{Corollary}

\newtheorem{lemma}[theorem]{Lemma}
\newtheorem{claim}[theorem]{Claim}
\newtheorem{proposition}[theorem]{Proposition}
\newtheorem{remark}[theorem]{Remark}
\theoremstyle{definition}
\newtheorem{defn}[theorem]{Definition}

\newtheorem{example}[theorem]{Example}
\theoremstyle{remark}
\title{Query-Based Selection of Optimal Candidates under the Mallows Model\footnote{Parts of the work will be presented at the IEEE Information Theory Workshop (ITW) 2023, Saint-Malo, France.}}

\author{
Xujun Liu\thanks{Department of Foundational Mathematics, Xi'an Jiaotong-Liverpool University, Suzhou, Jiangsu Province, 215123, China, xujun.liu@xjtlu.edu.cn} \and
Olgica Milenkovic\thanks{Department of Electrical and Computer Engineering, University of Illinois, Urbana-Champaign, Urbana, IL, 61801, USA, milenkov@illinois.edu} \and
George V. Moustakides\thanks{Department of Electrical and Computer Engineering, University of Patras, Rio, 26500, Greece, moustaki@upatras.gr}
 }

\begin{document}
\maketitle

\begin{abstract} We study the secretary problem in which rank-ordered lists are generated by the Mallows model and the goal is to identify the highest-ranked candidate through a sequential interview process which does not allow rejected candidates to be revisited. The main difference between our formulation and existing models is that, during the selection process, we are given a fixed number of opportunities to query an infallible expert whether the current candidate is the highest-ranked or not. If the response is positive, the selection process terminates, otherwise, the search continues until a new potentially optimal candidate is identified.  Our optimal interview strategy, as well as the expected number of candidates interviewed and the expected number of queries used, can be determined through the evaluation of well-defined recurrence relations. Specifically, if we are allowed to query $s-1$ times and to make a final selection without querying (thus, making $s$ selections in total) then the optimum scheme is characterized by $s$ thresholds that depend on the parameter $\theta$ of the Mallows distribution but are independent on the maximum number of queries.  
\end{abstract}

\section{Introduction}\label{intro}

The secretary problem, also known as the game of googol and the picky bride problem, was formally introduced by Gardner~\cite{G1, G2} and is considered a prototypical example in sequential analysis, optimization, and decision theory. It can be stated as follows: $N$ individuals are assumed to be ranked from best-to-worst without ties according to their qualifications. They apply for a ``secretary'' position, and are interviewed one by one, in a uniformly random order. When the $i^{\text{th}}$ candidate appears, one can only rank her/him with respect to the $i-1$ already interviewed individuals. At the time of the $i^{\text{th}}$ interview, the employer can make the decision to hire the person presented or continue with the interview process by rejecting the candidate; rejected candidates cannot be revisited at a later time. If only one selection is to be made, what selection strategy (i.e., stopping rule) maximizes the probability of selecting the best (highest ranked) candidate? 

The first published solution to the problem is due to Lindley~\cite{L1} in 1961 and is based on algebraic methods. Dynkin~\cite{D1} solved the problem in 1963 by viewing the selection process as a Markov chain. For $N$ large enough, the answer turns out to be surprisingly elegant and simple: the first $N/e$ candidates are automatically rejected ($e$ stands for the base of the natural logarithm) and the first candidate that outranks all previously seen candidates after that point is selected for an offer. This strategy ensures a probability of successful identification of the best candidate equal to $1/e$, provided that $N$ is allowed to go to infinity.  

The secretary problem has been extended in many directions. Examples include full information games~\cite{GM2}, the classical secretary problem on posets~\cite{FW1, GM1, GKMN1, P1} and in the context of matroid theory~\cite{BIK1, Soto}, as well as the Prophet inequality model~\cite{KS1}. For more extensions and a detailed history of the development of secretary problem, the interested reader is referred to~\cite{F1, F2}. In particular, an extension of the classical secretary problem, known as the \emph{Dowry problem with multiple choices} (henceforth, the Dowry problem), was studied by Gilbert and Mosteller in their seminal work~\cite{GM2}. In the Dowry problem, one is allowed to select $s \ge 1$ candidates during the interview process, and the criteria for success is that the selected group includes the optimal candidate. This review process can be motivated or extended in many different ways: For example, one may view the $s$-collection to represent candidates invited for a second round of interviews. 

For both the secretary and Dowry problem, the modeling assumption is that the candidates are presented to the evaluator uniformly at random. Nevertheless, it is often the case that the candidates are presented in an order that is nonuniform~\cite{CJMTUW1, jones2019, jones2020weighted, LM1, postdocpaper}. This consideration lends itself to a generalization of the secretary problem introduced by Jones in~\cite{jones2020weighted}. The paper~\cite{jones2020weighted} 
considers candidates arriving in an order dictated by a sample permutation from the Mallows distribution~\cite{mallows1957non}. In particular, when $N \to \infty$, Jones~\cite{jones2020weighted} showed that the optimal strategy for the classical secretary problem under the Mallows model (parametrized by $\theta > 0$) is as follows: (1) when $\theta < 1$, reject all but the last $j = \max(-1/\ln{\theta}, 1)$ candidates and select the next left-to-right maximum thereafter; and (2) when $\theta > 1$, reject the first $k$ candidates and select the next left-to-right maximum thereafter. Here, $k$ is a function of $\theta$ but independent on $N$. It is important to point out that the focus of the work addressed the secretary problem (with one selection), with recent extensions providing companion solutions for the postdoc problem~\cite{LM1, postdocpaper}, introduced by Dynkin in the 1980s~\cite{R1, V1}. 

Motivated by a recent line of problems considering learning problems with queries~\cite{ashtiani,dixie,arya,barna}, we introduce the problem of query-based sequential analysis under the Mallows model. In our setting, we make use of the Mallows distribution\footnote{Our results actually apply to a broader class of distributions represented by prefix-equivariant statistics, as will be apparent from our proofs.} and assume that the decision making entity has access to a limited number of queries to an infallible expert. When faced with a candidate identified by an exploration-exploitation procedure as the potentially optimal choice, an expert provides an answer of the form ``Best'' and ``Not the best.'' If the answer is ``Not the best,'' a new exploration-exploitation stage is initiated, with the potential of using another query at the end of the process. If the answer is ``Best,'' the sequential examination process terminates. Given a budget of $s-1$ queries, where $s$ is usually a relatively small positive integer, e.g., $2 \le s \leq 5$, the questions of interest is to find the optimal interview strategy, the optimal probability of success and the expected number of candidates interviewed or experts queried until success or termination. Note that we are allowed to make a final selection without querying experts after using all $s-1$ queries; thus, we have a budget of $s-1$ queries and can make at most s selections.
 
When applied on the random interview model, the above setting resembles the Dowry problem as both allow for selecting $s$ candidates. In the Dowry problem, one is allowed to make $s$ selections without the information about the global ranking of the selected candidates while in our query-based setting the expert is asked if the candidate is globally the best. Furthermore, it assumes that candidate lists are Mallows samples and even for the Dowry problem, such a distribution assumption was not considered in the past. 

For a Mallows distribution parametrized by $\theta>0$, the case $\theta >1$ corresponds to a decreasing trend in the quality of candidates while the case $0 < \theta < 1$ corresponds to an increasing trend in the quality of candidates. The case $\theta=1$ corresponds to the uniform distribution. In our query-based model, we are provided with $s-1$ ($s \ge 1$) opportunities to query an expert whether the current candidate is the best or not. If the best candidate is not identified after these queries we are still allowed to make a final selection without querying experts. In terms of the optimal probability of winning and the optimal strategy, the query-based model with $s-1$ queries is equivalent to the Dowry problem with $s$ selections, since (i) they both have in total a maximum budget of $s$ selections; (ii.1) if a selected candidate presented as the $i^{\text{th}}$ selection, where $1 \le i \le s-1$, is the globally best candidate then in the query-based model we accept this candidate and stop our search. In the Dowry model we select, save, and move on since we still have at least one selection left, which will not influence the result as we already picked the best candidate; (ii.2) if the $s^{\text{th}}$ selection is the best candidate, the two models are equivalent as we will save the selected candidate in both models; (iii.1) if a candidate at the $i^{\text{th}}$ selection time, where $1 \le i \le s-1$, is not the globally best candidate then in the query-based model we will be informed about this fact and will continue the search while in the Dowry model we will not have this information available but will still continue since there is at least one selection left; (iii.2) if a selected candidate at the $s^{\text{th}}$ selection point is not the best, then under both models we will save the candidate and terminate.

The optimal query and selection strategies for both our query-based model and the Dowry model depend on the value of the parameter $\theta$ of the Mallows model. For $N \to \infty$ and $\theta > 1$, the optimal strategy\footnote{Since there may exist more than one optimal strategies which can attain the optimal winning probability, we use ``an optimal strategy'' to refer to one of the optimal strategy and ``the optimal strategy'' to refer to all optimal strategies.} is an $s$-threshold $(k_1, \ldots, k_s)$-strategy (formally defined in Theorem~\ref{strategy}) with $0 \le k_1 \le \ldots \le k_s \not \to \infty$, where $k_i$, $1 \le i \le s$, is the threshold for the $i^{\text{th}}$ selection. In this setting, an optimal strategy is as follows: When making the $i^{\text{th}}$ selection, for each $1 \le i \le s$, we reject all candidates up until position $k_i$, then select the next left-to-right maxima (a candidate which is the best when compared with all examined candidates up to that point). For $0 < \theta < 1$, the optimal strategy is also an $s$-threshold $(k_1, \ldots, k_s)$-strategy; however, in this case, $ k_1 \le \ldots \le k_s \le N-1$ and $N-k_1 \not \to \infty$. 

Furthermore, let $N$ be a fixed positive integer; for each positive real number $\theta$, there exists a sequence of numbers $a_1(\theta), a_2(\theta),\ldots$ that depends only on $\theta$, such that $a_1(\theta) \ge a_2(\theta) \ge a_3(\theta) \ge \ldots$ and an optimal strategy for the proposed query-based model (with $s-1$ queries in total) is the $s$-threshold $(a_s(\theta), \ldots, a_1(\theta))$-strategy. In addition, the thresholds in the $(a_s(\theta), \ldots, a_1(\theta))$-strategy do not change with the value of $s$ if read from the right. For example, let $\theta>0$ be fixed, and assume that the optimal strategy for $s = 2$ is the $(a_2(\theta), a_1(\theta))$-strategy and an optimal strategy for $s = 3$ is the $(a_3(\theta), a_2(\theta), a_1(\theta))$-strategy; when $s$ increases from $2$ to $3$, our optimal strategy will only add $a_3(\theta)$ on the left and the parameter values at later positions in the strategy will not change with $s$. Even though a special case of our version of the problem exhibits similarities with the Dowry problem with multiple choices proposed by Gilbert and Mosteller~\cite{GM2}, there are essential structural and methodological differences, since in our case we allow for early stopping, address non-uniform (e.g., Mallows) candidate interview distributions and, most importantly, provide an exact (non-asymptotic) proof of the optimality of our scheme for any number of queries. 

An important combinatorial method to study sequential problems under nonuniform ranking models was developed in a series of papers by Fowlkes and Jones~\cite{FJ19}, and Jones~\cite{jones2019,jones2020weighted}. For consistency, we use some of the notation and definitions from Jones~\cite{jones2020weighted} but also introduce a number of new concepts and combinatorial  proof techniques. 
In particular, finding recurrence relations for more than one selection is significantly more challenging than for the secretary problem, and the optimal strategies differ substantially from the classical ones as our results include multiple thresholds for stopping. 

The paper is organized as follows. Section~\ref{sec:preliminaries} introduces the relevant concepts, terminology and models used throughout the paper. The same section also contains a number of technical lemmas that help in establishing our main results pertaining to the optimal selection strategies. An in-depth analysis of the exploration stages and the probabilities of success for the optimal selection processes under the Mallows distribution are presented in Section~\ref{mallows}. Numerical results for the exploration phase lengths and optimal winning probabilities versus $\theta$ are discussed in Section~\ref{numerical}. Results of our analysis for the expected number of questions (selections) used as well as the numerical results are presented in Section~\ref{expect}. 


\vspace{-3mm}
\section{Preliminaries}\label{sec:preliminaries}

The sample space is the set of all permutations of $N$ elements, i.e. the symmetric group $S_N$, with the underlying $\sigma$-algebra equal to the power set of $S_N$. The best candidate is indexed by $N$, the second-best candidate by $N-1,\ldots,$ and the worst candidate is indexed by $1$. The interview committee can accurately compare the candidates presented up to a certain time point, but cannot assess the quality of the future candidates. We also assume that there is a budget of $s-1$ queries ($s \ge 1$) to be made that produce an answer whether a current candidate is the globally best one or not, as well as a final selection that does not involved queries (thus, a total of $s$ selections). These modeling assumptins are equivalent to those of the Dowry problem with $s$ selections if one only considers the optimal probability of success and winning strategy. Still, there are differences in the expected interview times which are discussed in Section~\ref{expect}. 

Furthermore, unlike the standard model of the secretary and Dowry problem, our framework assumes that the candidates are presented (one-by-one, from the left) according to a permutation (order) dictated by the Mallows distribution $\mathcal{M}_{\theta}$, parametrized by a real number $\theta>0$. The probability of presenting a permutation $\pi \in S_N$ to the hiring committee equals 
$$f(\pi) = \theta^{c(\pi)} \bigg/  \sum\limits_{\pi \in S_N} \theta^{c(\pi)},$$
where $c: S_N \to \mathbb{N}$ equals the smallest number of adjacent transpositions needed to transform $\pi$ into the identity permutation $[12\,\cdots \,N]$. Equivalently, $c(\pi)$ equals the number of pairwise element inversions, and is also known as the \emph{Kendall} $\tau$ distance between the permutation $\pi$ and the identity permutation $[12\,\cdots \,N]$. Note that the notation for a permutation in square bracket form should not be confused with the notation for a set $[a,b]=\{{a,a+1,\ldots,b\}},\, b\geq a,$ and the meaning of the notation used will be clear from the context. 


As remarked upon in Section~\ref{intro}, the query-based model and the Dowry model in the uniform permutation selection setting are the same problem when considering the maximum probability of winning and a corresponding optimal strategy. Therefore, for simplicity, we present our results for the new Mallows model in the ``language'' of the Dowry problem with $s$ selections.

\vspace{-2mm}
\subsection{The $Q,Q^o,\bar{Q}$ probabilities and strike sets}

For a given permutation $\pi \in S_N$ drawn according to the Mallows model, we say that a strategy \emph{wins} if it correctly identifies the best candidate when presented with $\pi$. The notion of a \emph{prefix} is introduced to represent the current relative ordering of candidates. Given a permutation $\pi \in S_N$, the $k^{(th)}$ prefix of $\pi$, denoted by $\pi|_k,$ is a permutation in $S_k$ and it represents the relabelling of the first $k$ elements of $\pi$ according to their relative order. For example, if $\pi=[635124] \in S_6$ and $k=4$, then  $\pi|_4=[4231]$.

\vspace{-2mm}

\begin{defn}
Let $\sigma \in \bigcup\limits_{i = 1}^{N} S_i$ and assume that the length of the permutation, $|\sigma|$, equals $k$. 

(1) We say that $\pi \in S_N$ is {\em $\sigma$-prefixed} if $\pi|_k = \sigma$. For example, $\pi=[165243] \in S_6$ is $\sigma=[1432]$-prefixed.

(2) Given that $\pi$ is $\sigma$-prefixed, we say that $\pi$ is {\em $\sigma$-winnable} if accepting the prefix $\sigma$, i.e. if accepting the $|\sigma|^{\text{th}}$ candidate when the order $\sigma$ is encountered identifies the best candidate of the interview order $\pi$. More precisely, for $\sigma = [\sigma(1)\sigma(2) \cdots \sigma(k)]$, we have that $\pi$ is $\sigma$-winnable if $\pi$ is $\sigma$-prefixed and $\pi(k) = N$.
\end{defn}

\begin{defn}
A {\em left-to-right maxima} in a permutation $\pi \in S_N$ is a position whose value is larger than all values to the left of the position. For example, if $\pi = [423516] \in S_6$, then the first, fourth, and sixth position are left-to-right maxima.
\end{defn}

\vspace{-5mm}
\begin{defn}
We say that a permutation $\sigma \in \bigcup\limits_{i = 1}^{N} S_i$ is {\em eligible} if it ends in a left-to-right maxima or has length $N$. For example, let $N = 6$. Then both $[1324]$ and $[165243]$ are eligible.
\end{defn}

A permutation $\pi \in S_N$ is sampled from the Mallows model before the interview process. During the interview process, each entry of $\pi$ is presented one-by-one from the left; the relative ordering of the positions presented so far forms a prefix of $\pi$. That is the only information that can be used to decide whether to accept or reject the current candidate. Therefore, every strategy can be represented as a set of permutations of possibly different lengths that lead to an accept decision for the last candidate observed; such a set is called a {\em strike set}. More precisely, the selection process proceeds as follow: If the prefix we have seen so far is in the strike set, then we accept the current candidate and continue (if there is at least one selection left); if it does not belong to the strike set, we reject the current candidate and continue. For example, let $N=4$ and $s=1$. Then, the boxed set of permutations $A = \{[12], [213], [3124], [3214]\}$ in Figure~\ref{tree-1} is a strike set. The corresponding interview strategy may be summarized as follows: If the relative order of the candidates interviewed so far is in the set $A$, then accept the current candidate; otherwise, reject the current candidate. 

Since in our model there are $s$ selections, we make use of $s$-strike sets defined below.

\begin{defn}
A set $X \subseteq \bigcup_{j = 1}^{N} S_j$ is called an {\em $s$-minimal set} if it is impossible to have $s+1$ elements $\alpha_1, \alpha_2, \ldots, \alpha_{s+1} \in X$ such that $\alpha_{i+1}$ is a prefix of $\alpha_i$, for all $i \in \{1, 2, \ldots, s\}$.
\end{defn}

\begin{defn}
A set of permutations $A \subseteq \bigcup_{j = 1}^{N} S_j$ is called an {\em $s$-strike set} if it satisfies the following three conditions:

(1) It comprises prefixes that are eligible. 

(2) The set $A$ is $s$-minimal. The set $A$ may contain elements $\alpha_1, \alpha_2, \ldots, \alpha_{s}$ such that $\alpha_{i+1}$ is a prefix of $\alpha_i$, for all $i \in \{1, 2, \ldots, s-1\}$. In other words, based on an $s$-strike set one can make at most $s$ selections. 

(3) Every permutation in $S_N$ contains some element of $A$ as its prefix (i.e., given an $s$-strike set one can always make a selection based on its elements).
\end{defn}

A $1$-strike set corresponds to the valid strike set defined in the paper of Jones~\cite{jones2020weighted} when only one choice is allowed. We use the term strike set whenever $s$, the number of total selections, is clear from the context.

From the previous definition and the fact that we are allowed to make at most $s$ selections it follows that any optimal strategy for our problem can be represented by an $s$-strike set. For example, the set \{[1], [12], [213], [3124], [3214]\} in Figure~\ref{tree-1} is a $2$-strike set, which also represents an optimal strategy for the case $N = 4$ and $s=2$. See also Example~\ref{prefixtree}. Furthermore, for a permutation $\sigma$ of length $k$, where $1 \le k \le N$, and $i \in \{1, 2, \ldots, s\}$, we make extensive use of the following probabilities. 

\begin{defn}\label{probabilities}
Let $\sigma$ be a permutation of length $k$, where $1 \le k \le N$, and let $i \in \{1, 2, \ldots, s\}$. Define

$Q_i(\sigma)$, the probability of identifying the best
candidate with the strategy accepting the $k^{\text{th}}$ position
and using the best strategy thereafter \textbf{conditioned on}
the pre-selected interviewing order $\pi$ being $\sigma$-prefixed and
$i$ selections still being available when interviewing the candidate at position $k$.

$Q_i^o(\sigma)$, the probability of identifying the best
candidate with the best strategy after making a decision
for the $k^{\text{th}}$ position \textbf{conditioned on} the pre-selected interviewing order $\pi$ being $\sigma$-prefixed and $i$ selections still being available right after the interview of the $k^{\text{th}}$ candidate.

$\bar{Q}_i(\sigma)$, which equals $\bar{Q}_i(\sigma) = \max\{Q_i(\sigma), Q^o_i(\sigma)\}.$
\end{defn}

In words, $Q$ represents the probability of winning by accepting the current candidate while $Q^o$ is the probability of winning based on future selections in the interview process. In order to ensure the maximum probability of winning, an optimal strategy will examine two available choices, i.e. ``accept the current candidate'' or ``reject the current candidate and implement the best strategy in the future'' at each stage of the interview and select the one with a better chance of identifying the best candidate.

\begin{remark}\label{large}
Let $1 \le i \le s$. Consider an arbitrary permutation $\sigma$ of length $N$. If the last position of $\sigma$ is $N$, then $Q_i(\sigma) = 1$, $Q^o_i(\sigma) = 0$, and $\bar{Q}^o_i(\sigma) = 1$. If the last position of $\sigma$ is not $N$, then $Q_i(\sigma) = 0$, $Q^o_i(\sigma) = 0$, and $\bar{Q}^o_i(\sigma) = 0$, since selecting the last candidate will not result in success and the search will terminate after interviewing the last candidate. For a permutation $\sigma$ of length at least $N-i$, $Q_{i+1}(\sigma) = \ldots = Q_{s}(\sigma)$. Furthermore, for a permutation $\sigma$ of length at least $N-i$, $Q^o_{i}(\sigma) = Q_i^o(\sigma) = \ldots = Q_{s}^o(\sigma)$.
\end{remark}

We also need the following definitions.

\begin{defn}
Let $\sigma$ be a permutation of length $k \le N$. {\em The standard denominator} $SD(\sigma)$ of $\sigma$ equals
$$SD(\sigma)=\sum\limits_{\sigma\text{-prefixed } \pi \in S_N} \theta^{c(\pi)}.$$
Furthermore, $\text{Win}(\sigma)$ stands for the sum of the weights $\theta^{c(\pi)}$ over all $\sigma$-winnable permutations $\pi \in S_N$.
\end{defn}

The case $i = 1$ was first analyzed by Jones~\cite{jones2020weighted}, establishing that
\begin{equation}\label{basicq0}
Q_1(\sigma) = \frac{\sum\limits_{\sigma\text{-winnable }\pi \in S_N} \theta^{c(\pi)}}{\sum\limits_{\sigma\text{-prefixed }\pi \in S_N} \theta^{c(\pi)}} = \frac{Win(\sigma)}{SD(\sigma)}.
\end{equation}

\begin{defn}
The $\xor$ operation for two fractions $\frac{a}{b}$ and $\frac{c}{d}$ is defined as $\frac{a}{b} \xor \frac{c}{d} = \frac{a+c}{b+d}$. 
\end{defn}
Roughly speaking, the $\xor$ operation will be used to compute the probability of the union of two disjoint events from two disjoint sample spaces over a new sample space equal to the union of the sample spaces. It is important to point out that we do not cancel common divisors in the defining fractions for the probabilities $Q,Q^o, \bar{Q}$ until the final step. 

Each entry of pre-selected permutation $\pi$ is presented one-by-one from the left during the interview process, and the relative ordering of the already observed positions forms a prefix of $\pi$. This relative ordering changes with more candidates being interviewed and we need a means to describe this process.

\begin{defn}\label{children}
For each $\sigma$ of length $\ell-1$, where $\ell \le N$, we define $\lambda_j(\sigma)$, $1 \le j \le \ell$, to be the $\sigma$-prefixed permutation of length $\ell$ such that its last position has value $j$; the permutation is obtained by relabelling the $\ell$ positions of $\sigma$. For example, for $\sigma = [123],$ a permutation of length $3$, we have $\lambda_1(\sigma) = [2341], \lambda_2(\sigma) = [1342], \lambda_3(\sigma) = [1243]$ and $\lambda_4(\sigma) = [1234]$. 
\end{defn}

Let $\sigma$ be a permutation of length $1 \le k \le N$ with $Q_i(\sigma), Q^o_i(\sigma), \bar{Q}_i(\sigma)$ defined as above, and $1 \leq i \leq s$ selections available right before processing the $k^{\text{th}}$ candidate of a $\sigma$-prefixed permutation. For each $1 \le i \le s$, if the $k^{\text{th}}$ position of $\sigma$ is selected, then the number of selections available decrease by one; if the $k^{\text{th}}$ candidate is rejected, then the number of selections available does not change. When the number of available selections becomes zero or all candidates are examined, the process terminates.

After making a decision on the $|\sigma|^{\text{th}}$ candidate, the interviewer examines the next applicant while the relative order of the interviewed candidates changes to one of $\lambda_1(\sigma), \ldots, \lambda_{k+1}(\sigma)$. An optimal strategy involves making a decision with the largest probability of winning when encountering each of the $\lambda_1(\sigma), \ldots, \lambda_{k+1}(\sigma)$:
\begin{equation}\label{basicqio}
Q^o_i(\sigma) = \bar{Q}_i(\lambda_1(\sigma)) \xor \cdots \xor \bar{Q}_i(\lambda_{k+1}(\sigma)).
\end{equation}

Proposition~\ref{keyidea} provides a way to write $Q_i(\sigma)$ using $Q$ probabilities with smaller subscripts (i.e., $Q_1(\sigma)$ and $Q^o_{i-1}(\sigma)$). This simple result is heavily used in the proofs to come.
\begin{proposition}\label{keyidea}
For $i \in \{2, \ldots, s\}$,
\begin{equation}\label{qqoqbar}
Q_i(\sigma) = Q_{1}(\sigma) + Q_{i-1}^o(\sigma).
\end{equation}
\end{proposition}

\begin{proof}
Equation~\eqref{qqoqbar} holds since there are two (disjoint) events that ensure winning after examining the current candidate, i.e., (a) the current candidate is the best and (b) the current candidate is not the best but we identify the best candidate at a later time with a best strategy after rejecting the current candidate. In the first case, the probability of successfully identifying the best candidate is $Q_1(\sigma)$; in the second case, the number of available selections decreases by one and the corresponding probability is $Q_{i-1}^o(\sigma)$.
\end{proof}

Following an approach suggested by Jones~\cite{jones2020weighted}, we make use \textit{prefix trees} which naturally capture the inclusion relationships between prefixes of permutations. The concept is best described by an illustrative example, shown in Figure~\ref{tree-1} for $S_4$. The correspondence between sub-trees/sub-forests is crucial for the proof of Lemma~\ref{qi} and~\ref{only length}. 

\begin{defn}
Let $T$ be the tree capturing the inclusion relationships between prefixes of permutations of length at most $N$. In other words, for $V$ being the collection of all permutations of length at most $N$, we let $T = (V, E)$ be such that if $\sigma, \tau \in V$ and $\sigma$ is a prefix of $\tau$ with $|\sigma| = |\tau| - 1$, then we have an edge $\sigma \tau \in E$. We define $\bar{T}(\sigma)$ to be the subtree in $T$ comprising $\sigma$ and its descendants and let $T^o(\sigma) = \bar{T}(\sigma) - \sigma$ be the forest obtained by deleting $\sigma$ from $\bar{T}(\sigma)$. 
\end{defn}

In Figure~\ref{tree-1}, if $\sigma = [12]$, then $\bar{T}(\sigma)$ is the subtree induced by the vertices $$\{[12], [123], [132], [231], [1234], [1243], [1342], [2341], [1324], [1423], [1432], [2431], [2314], [2413], [3412], [3421]\}$$  and  $T^o(\sigma)$ is the forest induced by the vertices $$\{[123], [132], [231], [1234], [1243], [1342], [2341], [1324], [1423], [1432], [2431], [2314], [2413], [3412], [3421]\}.$$

Let $F_1$ be the sub-forest obtained by deleting the vertex $[12]$ in the tree induced by $[12]$ and its descendants, let $F_2$ be the sub-forest obtained by deleting the vertex $[21]$ in the tree induced by $[21]$ and its children. Then, there is a bijection between $F_1$ and $F_2$ which preserves all the probabilities used for evaluating the selection strategies.

\begin{defn}
We say that a prefix $\sigma$ is {\em type $i$-positive} if $Q_i(\sigma) \ge Q^o_i(\sigma)$ and {\em type $i$-negative} otherwise, for $i \in \{1, \ldots, s\}$.
\end{defn}

We show next that the probabilities $Q_i^o, Q_i, \bar{Q}_i$ for each $i \in \{1, \ldots, s\}$ can be calculated (pre-calculated) using a sequential procedure (backward induction). Based on this result, we will find the winning probability in Section~\ref{mallows} by solving a few well-defined recurrence relation.

\begin{proposition}\label{probs-qiqio}
Let $\tau$ be any permutation of length $k$, where $1 \le k \le N$. The probabilities $Q_i^o(\tau), Q_i(\tau), \bar{Q}_i(\tau)$ for each $i \in \{1, \ldots, s\}$ can be computed recursively.
\end{proposition}

\begin{proof}
In order to compute the probabilities $Q_i^o(\tau), Q_i(\tau)$ for each $i \in \{1, \ldots, s\}$, we use a double-induction on the subscript $i$ and the length of the prefix. 

\textbf{Base case for outer induction on $i$:} We first establish the base case for $i = 1$. By Remark~\ref{large}, the permutations of length $N$ are type $1$-positive, which establishes the base case for the induction on the length of a prefix. More precisely, for a permutation $\tau$ of length $N$, if $\tau(N) = N$ then $Q_1(\tau) = \bar{Q}_1(\tau) = 1$ and $Q^o_1(\tau) = 0$; if $\tau(N) < N$ then $Q_1(\tau) = Q^o_1(\tau) = \bar{Q}_1(\tau) = 0$. 

Assume that the probabilities $Q_1, Q^o_1, \bar{Q}_1$ for permutations of length longer than $k$, $1 \le k \le N-1$, are already known. We show that $Q_1(\tau), Q_1^o(\tau)$, and $\bar{Q}_1(\tau)$ can then also be determined for a length-$k$ permutation $\tau$. By Equation~\eqref{basicq0}, the value of $Q_1(\tau)$ can be obtained by finding a fraction with denominator equal to the sum of $\theta^{c(\pi)}$ over all $\pi \in S_N$ that are $\tau$-prefixed (i.e., $SD(\tau)$) and the numerator equal to the sum of $\theta^{c(\pi)}$ over all $\pi \in S_N$ that are $\tau$-winnable; those values are available since $\theta$ and the statistic $c$ (Kendall distance in our model) are known. By Equation~\eqref{basicqio}, the probability $Q^o_1(\tau)$ can be obtained from $Q^o_1(\tau) = \xor_{j = 1}^{k+1} \bar{Q}_1(\lambda_{j}(\tau))$. Since each $\lambda_{j}(\tau)$ has length larger than that of $\tau$, each of the $ \bar{Q}_1(\lambda_{j}(\tau))$ is already available according to the inductive hypothesis. The $\bar{Q}_1(\tau)$ probabilities can be determined from $\bar{Q}_1(\tau) = \max\{Q_1(\tau), Q_1^o(\tau)\}$, where $Q_1(\tau)$ and $Q_1^o(\tau)$ are known.


\textbf{Main proof following the base case:} Assume now that we know the $Q_q, Q^o_q, \bar{Q}_q$ probabilities for each $1 \le q \le i-1$, $i \in \{2, \ldots, s\}$, and for every permutation in $\bigcup\limits_{j = 1}^{N} S_j$. We prove the claimed result for $i$. The probabilities $Q_i, \bar{Q}_i$ of prefixes of length $N$ take either the value $1$ or $0$ depending on whether the last position has value $N$, while the probabilities $Q^o_i$ are all $0$. This serves as the base case for the inner induction argument on the length of the prefixes. 

Let $\tau$ be a permutation of length $k$, where $1 \le k \le N-1$. Given the probabilities $Q_i, Q_i^o, \bar{Q}_i$ for prefixes of lengths greater than $k$, the $Q_i^o(\tau)$ probabilities for prefixes $\tau$ of length $k$ can be obtained via  $Q_i^o(\tau) = \xor_{j = 1}^{k+1} \bar{Q}_i(\lambda_{j}(\tau))$ from Equation~\eqref{basicqio}; the probabilities $\bar{Q}_i(\lambda_{j}(\tau))$ are known by the inductive hypothesis. Moreover, we can find $Q_i(\tau)$ by Proposition~\ref{keyidea}, i.e., $Q_i(\tau) = Q_1(\tau) + Q_{i-1}^o(\tau)$, where $Q_1(\tau)$ is known from the base case analysis and $Q_{i-1}^o(\tau)$ is available based on the inductive hypothesis. Finally, $\bar{Q}_i(\tau)$ can be found using the definition $\bar{Q}_i(\tau) = \max\{Q_i(\tau), Q^o_i(\tau) \}$.
\end{proof}

Recall that by Equation~\eqref{basicq0}, $Q_1(\sigma)$ can be written as a fraction with denominator $SD(\sigma)$ and numerator equal to the sum of $\theta^{c(\pi)}$ over all $\pi$ that are $\sigma$-winnable. In Propositions~\ref{easyexpansion} below we show that the probabilities $Q_i(\sigma)$, where $2 \le i \le s$, and $Q^o_i(\sigma)$, where $1 \le i \le s$, can also be expressed as fractions with the standard denominator $SD(\sigma)$.

\begin{proposition}\label{easyexpansion}
For each $1 \le i \le s$ and a permutation $\sigma$ of length $\ell$ with $\ell \le N-1-i$, there exists a collection of $\sigma$-prefixed permutations $\Gamma_{\sigma, i}$ such that each $\mu \in \Gamma_{\sigma, i}$ is of length larger than $|\sigma|$ and type $i$-positive. Furthermore, the set $\Gamma_{\sigma, i}$ is $1$-minimal, and $Q^o_i(\sigma) = \xor_{\mu \in \Gamma_{\sigma, i}} Q_i(\mu)$, i.e.,
\begin{equation}\label{qiqo-1}
Q^o_i(\sigma) \cdot SD(\sigma) = \sum\limits_{\mu \in \Gamma_{\sigma, i}} Q_i(\mu) \cdot SD(\mu) \quad \text{ and } \quad SD(\sigma) = \sum\limits_{\mu \in \Gamma_{\sigma,i}} SD(\mu). 
\end{equation}
Furthermore,
\begin{equation}\label{qisd-1}
Q_i(\sigma) \cdot SD(\sigma) = Q_1(\sigma) \cdot SD(\sigma) + \sum\limits_{\mu \in \Gamma_{\sigma, i-1}} Q_{i-1}^o(\mu) \cdot SD(\mu).
\end{equation}
\end{proposition}

\begin{proof}
The case $i = 1$ in Equation~\eqref{qiqo-1} was analyzed in~\cite{jones2020weighted}. Since $Q_2(\sigma) = Q_1(\sigma) + Q^o_1(\sigma)$, there is a set $\Gamma_{\sigma,1}$ such that
$$Q_2(\sigma) \cdot SD(\sigma) = Q_1(\sigma) \cdot SD(\sigma) + Q^o_1(\sigma) \cdot SD(\sigma) =  Q_1(\sigma) \cdot SD(\sigma) + \sum\limits_{\mu \in \Gamma_{\sigma,1}} Q_1(\mu) \cdot SD(\mu),$$
where $\Gamma_{\sigma,1}$ is $1$-minimal and consists of type $1$-positive permutations of length $> |\sigma|$. 

After making a decision on the $|\sigma|^{\text{th}}$ candidate, an optimal strategy will examine the children of $\sigma$  in the prefix tree, i.e. $\lambda_{1}(\sigma), \ldots, \lambda_{\ell+1}(\sigma)$, and then make a decision that leads to the largest probability of winning. We present the following algorithm that prove the part of the proposition pertaining to $Q^o_i(\sigma)$. 

\begin{itemize}
\item[Initialization step:] Let $\Gamma_i = \varnothing$ and $B = \{\lambda_{1}(\sigma), \ldots, \lambda_{\ell+1}(\sigma)\}$. 

We repeat the Main step until the process terminates.

\item[Main step:] Check if $B = \varnothing$; if yes, stop and return the set $\Gamma_i$; if no, then do the following: Pick an arbitrary permutation $\phi \in B$, say of length $q$, with $|\sigma| < q \le N$; check if $\phi$ is both eligible and type $i$-positive ($Q_i(\phi) \ge Q_i^o(\phi)$); if yes, set $\Gamma_i = \Gamma_i \cup \phi$ and $B = B - \phi$; if no, do not update $\Gamma_i$ and let $B = (B - \phi) \cup \bigcup\limits_{j = 1}^{q+1} \lambda_{j}(\phi)$. Note that the probabilities $Q_i(\phi)$ and $Q_i^o(\phi)$ are known by Proposition~\ref{probs-qiqio}.
\end{itemize}

Since the permutations of length $N$ are type $i$-positive for each $1 \le i \le s$, the algorithm eventually terminates. By the criteria on the main step of the algorithm, it will produce a set $\Gamma_i$ of type $i$-positive eligible permutations that is $1$-minimal and each of the $\gamma \in \Gamma_i$ has length larger than $|\sigma|$. At the end of the process, $B$ is an empty set. This follows from two observations.

\textbf{Observation (i):} There is no pair of elements $\alpha,\beta \in \Gamma_i$ such that $\alpha$ is a prefix of $\beta$, i.e., $\Gamma_i$ contains $1$-minimal prefixes, since otherwise the sub-forest $T^o(\alpha)$ will not be processed by the algorithm and it will be impossible for $\beta$ to be selected for inclusion in $\Gamma_i$.

\textbf{Observation (ii):} Since we choose a permutation only if it is type $i$-positive and eligible, every permutation in $\Gamma_i$ is type $i$-positive and eligible. 

Furthermore, by the main step of the algorithm and the induction hypothesis, 
$$Q^o_i(\sigma) \cdot SD(\sigma) = \sum\limits_{\gamma \in \Gamma_i} Q_i(\gamma) \cdot SD(\gamma) \quad \text{ and } \quad SD(\sigma) = \sum\limits_{\gamma \in \Gamma_i} SD(\gamma).$$
Equivalently, if we divide by $SD(\sigma)$ on both sides, we obtain
$$
Q^o_i(\sigma) = \xor_{\gamma \in \Gamma_i} Q_i(\gamma).
$$

To prove the corresponding formula for $Q_{i+1}(\sigma)$, note that $Q_{i+1}(\sigma) = Q_1(\sigma) + Q^o_i(\sigma)$ and invoke the result of~\eqref{qiqo-1} for $Q_i^o(\sigma)$. 
\end{proof}

Given the probabilities $Q_i^o, Q_i, \bar{Q}_i$ for all permutations and $i \in \{1, \ldots, s\}$, we describe next a procedure for finding an optimal strategy and its corresponding strike set.

\begin{theorem}\label{winprob-2}
There exists an $s$-strike set $A$ which can be partitioned as $A_{s} \cup \cdots \cup A_1,$  where each $A_i$ is a set of type $i$-positive $1$-minimal permutations, $1 \le i \le s$. The maximum probability of winning equals
$\xor_{\sigma \in A_{s}} Q_{s}(\sigma)$. Expressed in terms of the probability $Q_1$, the maximum probability reads as 
\begin{equation}\label{winningprob}
\sum\limits_{\sigma \in A} Q_1(\sigma) \cdot SD(\sigma) \bigg/  \sum\limits_{\pi \in S_N} \theta^{c(\pi)}.
\end{equation}
\end{theorem}

\begin{proof}
The optimal winning probability is $\bar{Q}_{s}([1])$. We start by checking whether $Q_s([1]) \ge Q_s^o([1])$. 

\textbf{Case 1:} $Q^o_s([1]) > Q_s([1])$. Then the strike set $A_s$ corresponds to the set $\Gamma_{[1],s}$ of Proposition~\ref{easyexpansion} and the winning probability equals $Q^o_s([1])$. By Equation~\eqref{qiqo-1} in Proposition~\ref{easyexpansion}, we need to examine each of the permutations in $A_s$ in order to find $Q^o_s([1])$. 

\textbf{Case 2:} $Q_s([1]) \ge Q_s^o([1])$. Then the strike set $A_s = \{[1]\}$ and the winning probability equals $Q_s([1])$. 

For both Case 1 and Case 2, we apply Equation~\eqref{qisd-1} to each $\mu \in A_s$ and then find $Q^o_{s-1}(\mu)$ for each $\mu \in A_s$. We apply Proposition~\ref{easyexpansion} again and obtain a strike set $A_{s-1}$. We can 
use this process to find $A_{s-2}$, then $A_{s-3}$, $\ldots$, and finally $A_1$. Furthermore, it follows that each $A_i$ is type-i-positive and 1-minimal, the set $A$ is an $s$-strike set, and Equation~\eqref{winningprob} holds.
\end{proof}

\begin{figure}
\begin{center}
  \includegraphics[scale=0.65]{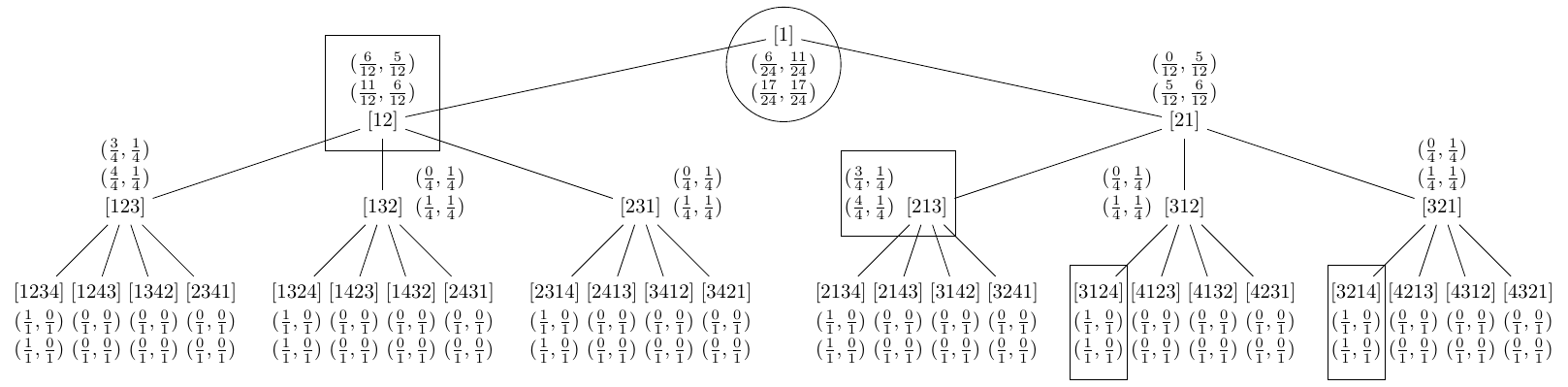}
 \caption{The prefix tree, $Q_1, Q^o_1, Q_2, Q^o_2$ probabilities , and a $2$-strike set when $N=4$ and $\theta = 1$.}\label{tree-1}
\end{center}
\vspace{-10mm}
\end{figure}

\vspace{-5mm}
\subsection{Properties of the $Q,Q^o,\bar{Q}$ probabilities}

\begin{defn}\label{prefixequivalentdef}
A statistic $c$ is said to be {\em prefix-equivariant} if it satisfies $c(\pi) - c(g_\tau \cdot \pi) = c([12 \cdots k]) - c(\tau)$ for all prefixes $\tau$ and all $\pi \in \bar{T}([12 \cdots k])$, where $k$ is the length of $\tau$.
\end{defn}

Intuitively, the condition $c(\pi) - c(g_\tau \cdot \pi) = c([12 \cdots k]) - c(\tau)$ enforces the statistic $c$ to have the property that permuting the first $k$ entries does not create or remove any ``structures'' counted by the statistic $c$ that exist at positions larger than $k$. The condition also ensures many useful properties for the probabilities $Q, Q^o, \bar{Q}$, including invariance under local changes (say, permuting the elements in a prefix). Prefix equivalence will be extensively used in the proofs of the theorems and lemmas to follow. It is straightforward to check that the Kendall statistic is prefix-equivariant.

\begin{defn}
Define $g_\tau$ to be an action on the symmetric group that arranges (permutes) a prefix $\sigma = [12 \cdots k]$ to some other prefix $\tau$ of the same length $k$. This action can be extended to $\bar{T}(\sigma)$ and is denoted by $g_\tau \cdot \pi$: It similarly permutes the first $k$ entries and fixes the remaining entries of $\pi \in \bar{T}(\sigma)$. It is easy to see that $g_\tau$ is a bijection from $\bar{T}(\sigma)$ to $\bar{T}(\tau)$.
\end{defn}

\begin{example}
Let $N = 6$, $\tau = [132]$, and $\pi = [245361]$. Clearly, $\pi$ is $[123]$-prefixed. Then $g_{\tau} \cdot \pi = [254361]$.
\end{example}








We show in Lemma~\ref{qi} that the probabilities $Q_i(\sigma), Q^o_i(\sigma), \bar{Q}_i(\sigma)$ only depend on the length and the value of the last position of $\sigma$, where $1 \le i \le s$.

\begin{lemma}\label{qi}
Suppose $c$ is a prefix-equivariant statistic. For all permutations $\tau$ of length $k \le N-1$, the probabilities $Q_i$ for each $i \in \{1, \ldots, s\}$ are preserved under the restricted bijection $g_{\tau}: T^o(12 \cdots k) \to T^o(\tau)$. Furthermore, if $\tau$ is eligible then $Q_i([12 \cdots k]) = Q_i(\tau)$. Consequently, for $\sigma \in T^o(12 \cdots k)$, 

(a) The probabilities $Q_i^o(\sigma)$ are preserved by $g_{\tau}$;

(b) The probabilities $\bar{Q}_i(\sigma)$ are preserved by $g_{\tau}$;

(c) If $\sigma$ and $\tau$ are eligible, we have that $\sigma$ is type $i$-positive if and only if $g_{\tau} \cdot \sigma$ is type $i$-positive. 

Additionally, for $\sigma = [12 \cdots k]$ and eligible permutation $\tau$ of length $k$, the statements (a),(b),(c) hold.
\end{lemma}

\begin{proof}
The proof proceeds by induction on the subscript $i$ of the probabilities $Q_i, Q^o_i, \bar{Q}_i$. The case $i = 1$ was analyzed in Theorem 3.5 of~\cite{jones2020weighted}. Assume that the result holds for the probabilities $Q_m, Q^o_m, \bar{Q}_m$ with $m \le i-1$. We next prove the claimed result for $Q_i, Q^o_i, \bar{Q}_i$, where $1 \le i \le s-1$. 

Let $\sigma \in T^o(12 \cdots k)$. We have 
$$Q_i(g_{\tau} \cdot \sigma) = Q_1(g_{\tau} \cdot \sigma) + Q_{i-1}^o(g_{\tau} \cdot \sigma) = Q_1(\sigma) + Q_{i-1}^o(\sigma) = Q_i(\sigma).$$ 
If $\tau$ is eligible then we apply the argument to the restricted bijection $T^o([12 \cdots (k-1)]) \to T^o(\tau|_{k-1})$. 

\begin{claim}\label{Q^o_i}
For $\sigma$ of length $k < \ell \le N$, $Q^o_i(g_{\tau} \cdot \sigma) = Q_i^o(\sigma)$. 
\end{claim}
\begin{proof}
We use induction on the length of $\sigma$. When $\sigma$ has length $N$ it holds that $Q^o_i(g_{\tau}\cdot \sigma) = Q^o_i(\sigma) = 0 $.  Assume now that statement (a) holds for prefixes of length at least $\ell+1$. We next present an argument for the case when $\sigma$ is of length $\ell$, where $k < \ell \le N-1$.

By the already proved result for the probability $Q_i$ and the induction hypothesis, the probabilities $Q_i$ and $Q^o_i$ for a permutation $\mu$ of length larger than $\ell$ only depend on the length of $\mu$ and the value of the last position of $\mu$. By the Main Step of the algorithm described in the proof of Proposition~\ref{easyexpansion}, if we process $\sigma$ and end up obtaining a set $\Gamma_i = \{\gamma_1, \ldots, \gamma_{r}\}$, then when we process $g_{\tau} \cdot \sigma$ we end up obtaining the set $\Gamma_i' = \{g_{\tau} \cdot \gamma_1, \ldots, g_{\tau} \cdot \gamma_{r}\}$. Therefore, by Proposition~\ref{easyexpansion}, 
$$
Q_i^o(g_{\tau} \cdot \sigma) = Q_i(g_{\tau} \cdot \gamma_1) \xor Q_i(g_{\tau} \cdot \gamma_2) \xor \cdots \xor Q_i(g_{\tau} \cdot \gamma_r)
$$
$$
= Q_i(\gamma_1) \xor Q_i(\gamma_2) \xor \cdots \xor \cdot Q_i(\gamma_r) = Q^o_i(\sigma). \qedhere
$$
\end{proof}

By Claim~\ref{Q^o_i}, statement (a) is true. Claims (b) and (c) can be established from the previous results and the fact $\bar{Q}_i(\sigma) = \max\{Q_i(\sigma), Q^o_i(\sigma)\}$. This completes the main part of the proof.

The statements (a), (b), and (c) hold for $\sigma = [12 \cdots k]$ and a permutation $\tau$ which is eligible and of length $k$, since we can apply the above argument to $T^o(12 \cdots (k-1)) \to T^o(\tau|_{k-1})$.
\end{proof}

We prove next that $Q^o_i(\sigma)$ depends on the length of $\sigma$ but not on the value of the last position of $\sigma$. 

\begin{lemma}\label{only length}
For all $1 \le i \le s$, the probability $Q^o_i(\sigma)$ only depends on the length of $\sigma$.
\end{lemma}

\begin{proof}
We know $Q_i^o(\sigma) = 0$ for every $\sigma$ of length $N$. Let $\sigma'$ and $\sigma''$ be two permutations of length $k-1$, where $k \le N$. For each permutation $\sigma$ of length $k-1$, we define $\sigma_j$, $1 \le j \le k$, as in Definition~\ref{children}. Let $\phi = [12 \cdots (k-1)]$. By Lemma~\ref{qi}, and using the bijections $g_{\sigma'}: T^o(12 \cdots (k-1)) \to T^o(\sigma')$ and $g_{\sigma''}: T^o(12 \cdots (k-1)) \to T^o(\sigma'')$, we have 
$$Q^o_i(\sigma') = \bar{Q}_i(\lambda_{1}(\sigma')) \xor \cdots \xor \bar{Q}_i(\lambda_{k}(\sigma')) = \bar{Q}_i(\lambda_{1}(\phi')) \xor \cdots \xor \bar{Q}_i(\lambda_{k}(\phi'))$$
$$ = \bar{Q}_i(\lambda_{1}(\sigma'')) \xor \cdots \xor \bar{Q}_i(\lambda_{k}(\sigma'')) = Q^o_i(\sigma''). \qedhere$$
\end{proof}

In order to simplify our exposition, in Lemma~\ref{relation-2} and Corollary~\ref{cor-2}, we change the notation and let $Q_i(\sigma)$, $Q_i^o(\sigma)$, $\bar{Q}_i(\sigma)$ stand \emph{only} for the \emph{numerators} in the definition of the underlying probabilities, each with respect to the standard denominator $SD(\sigma)$. All equalities involving the changed probability notations hold when the original denominators agree.



\begin{lemma}\label{relation-2}
Let $\sigma = [12 \cdots (k-1)]$ and define $\bar{Q}_0 = 0$ for any permutation. For $1 \le i \le s$, one has 
$$Q_i^o(\sigma) = \bar{Q}_i(\lambda_{k}(\sigma)) + Q_i^o(\lambda_{k}(\sigma)) \cdot \sum\limits_{j=1}^{k-1} \theta^{c(\lambda_{j}(\sigma)) - c(\lambda_{k}(\sigma))} \quad \text{and} \quad Q_i(\sigma) = \bar{Q}_{i-1}(\lambda_{k}(\sigma)) + Q_i(\lambda_{k}(\sigma)) \cdot \sum\limits_{j=1}^{k-1} \theta^{c(\lambda_{j}(\sigma)) - c(\lambda_{k}(\sigma))}.$$
\end{lemma}

\begin{proof}
The case $i=1$ was proved in Theorem 3.6 of~\cite{jones2020weighted}. We first consider $Q_i^o(\sigma)$. Note that $\sigma$ has $k$ children ($\lambda_{1}(\sigma), \ldots, \lambda_{k}(\sigma)$) in the prefix tree. The permutation $\lambda_{k}(\sigma)$ itself is an eligible child thus $\bar{Q}_i(\lambda_{k}(\sigma))$ is the optimal probability for the subtree rooted at $\lambda_{k}(\sigma)$. The subtrees beneath each of the permutations $\lambda_{i}(\sigma)$, $i \in \{1, \ldots, k-1\},$ are isomorphic to $T^o(\lambda_{k}(\sigma))$ via the bijections $g_{\lambda_{i}(\sigma)}$. For each $\lambda_{i}(\sigma)$-prefixed $\pi' \in S_N$, we have to account for a factor of $\theta^{c(\lambda_{i}(\sigma))-c(\lambda_{k}(\sigma))}$. This is due to the fact that $\pi'$ corresponds to a $\lambda_{k}(\sigma)$-prefixed permutation of $\pi$ via $g_{\lambda_{i}(\sigma)},$ such that $\pi' = g_{\lambda_{i}(\sigma)} \cdot \pi$ and $i \in \{1, \ldots, k-1\}$. 

Since $c$ is prefix-equivariant, 
$$\theta^{c(\pi')} = \theta^{c(\pi) - (c(\lambda_{i}(\sigma)) - c(\lambda_{k}(\sigma)))}.$$ 
Therefore, 
$$Q_i^o(\sigma) = \bar{Q}_i(\lambda_{1}(\sigma)) + \ldots + \bar{Q}_i(\lambda_{k-1}(\sigma)) + \bar{Q}_i(\lambda_{k}(\sigma)) $$ $$= Q^o_i(\lambda_{1}(\sigma)) + \ldots + Q^o_i(\lambda_{k-1}(\sigma)) + \bar{Q}_i(\lambda_{k}(\sigma)) = Q_i^o(\lambda_{k}(\sigma)) \cdot \sum\limits_{j=1}^{k-1} \theta^{c(\lambda_{j}(\sigma)) - c(\lambda_{k}(\sigma))} + \bar{Q}_i(\lambda_{k}(\sigma)).$$
By Equation~\eqref{qqoqbar}, we have 
$$Q_i(\sigma) = Q_1(\sigma) + Q_{i-1}^o(\sigma) = Q_1(\lambda_{k}(\sigma))  \cdot \sum\limits_{j = 1}^{k-1} \theta^{c(\lambda_{j}(\sigma)) - c(\lambda_{k}(\sigma))} + \bar{Q}_{i-1}(\lambda_{k}(\sigma)) + Q_{i-1}^o(\lambda_{k}(\sigma)) \cdot \sum\limits_{j = 1}^{k-1} \theta^{c(\lambda_{j}(\sigma)) - c(\lambda_{k}(\sigma))}$$
$$ = Q_{i}(\lambda_{k}(\sigma)) \cdot \sum\limits_{j = 1}^{k-1} \theta^{c(\lambda_{j}(\sigma)) - c(\lambda_{k}(\sigma))}+ \bar{Q}_{i-1}(\lambda_{k}(\sigma)). \qedhere$$
\end{proof}

We observe that $Q^o_i(\sigma) \ge Q_{i-1}^o(\sigma)$, $Q_i(\sigma) \ge Q_{i-1}(\sigma)$, $\bar{Q}_i(\sigma) \ge \bar{Q}_{i-1}(\sigma)$, and $Q_i(\sigma) \ge Q_{i-1}^o(\sigma)$ hold true for every $\sigma \in \bigcup\limits_{k=1}^{N}S_k$ and $1 \le i \le s$. By Lemma~\ref{relation-2}, we show in Corollary~\ref{cor-2} that if an eligible permutation is negative then all eligible permutations of shorter length are negative as well.

\begin{corollary}\label{cor-2}
For increasing prefixes $\sigma = [12 \cdots (k-1)]$ and $\lambda_{k}(\sigma) = [12 \cdots k]$, we have that if $\lambda_{k}(\sigma)$ is type $i$-negative then $\sigma$ is type $i$-negative, where $1 \le i \le s$.
\end{corollary}

\begin{proof}
The case $i=1$ was established in Corollary 3.7 of~\cite{jones2020weighted}. Suppose that $\sigma$ is type $i$-negative, i.e., such that $Q^o_i(\lambda_{k}(\sigma))>Q_i(\lambda_{k}(\sigma))$. By Lemma~\ref{relation-2} and $\bar{Q}_{i}(\lambda_{k}(\sigma)) \ge \bar{Q}_{i-1}(\lambda_{k}(\sigma))$, 
$$Q_i^o(\sigma) = Q_i^o(\lambda_{k}(\sigma)) \cdot \sum\limits_{j=1}^{k-1} \theta^{c(\lambda_{j}(\sigma)) - c(\lambda_{k}(\sigma))} + \bar{Q}_i(\lambda_{k}(\sigma)) >  Q_i(\lambda_{k}(\sigma)) \cdot \sum\limits_{j=1}^{k-1} \theta^{c(\lambda_{j}(\sigma)) - c(\lambda_{k}(\sigma))} + \bar{Q}_{i-1}(\lambda_{k}(\sigma)) = Q_i(\sigma). $$
\end{proof}
The probabilities $Q,Q^o,\bar{Q}$ henceforth refer to their original definition (with the denominators included). In words, Corollary~\ref{cor-2} asserts that each $Q_i(k) - Q_i^o(k)$ is a non-decreasing function of $k$. By Lemma~\ref{qi} and~\ref{only length}, if $\sigma$ is eligible, then the probabilities $Q_i(\sigma), Q_i^o(\sigma), \bar{Q}_i(\sigma)$ only depend on its length. Let $Q_i(k)$ denote the probability $Q_i$ of eligible permutations of length $k$, where $1 \le i \le s$. The probabilities $Q^o_i(k)$ and $\bar{Q}_i(k)$, where $1 \le i \le s$, are defined similarly.


\section{The optimal strategy}

By the proof of Theorem~\ref{winprob-2}, we start with checking if $Q_s([1]) \ge Q^o_s([1])$ while there are $s$ selections left. For the $i^{\text{th}}$ choice ($j = s+1-i$ selections left), $1 \le i \le s$, by the algorithm described in Proposition~\ref{easyexpansion}, we check if $\sigma$ is eligible and type $j$-positive, i.e., if $Q_j(\sigma) \ge Q^o_j(\sigma)$; if yes, we accept the current candidate and continue to the next selection (if there is one is left); if no, we reject the current candidate and continue our search; if there are no selections left, we terminate the process. By Corollary~\ref{cor-2}, we know each $Q_j(k) - Q_j^o(k)$ is a non-decreasing function of $k$, which allows us to formulate the optimal strategy.

\begin{theorem}\label{strategy}
Suppose that the probability distribution on $S_N$ is governed by a prefix-equivariant statistic (which includes the Kendall statistic) $c$.  For each fixed $\theta > 0$, an optimal strategy for our problem with $s$ selections is a positional $s$-thresholds strategy, i.e., there are $s$ numbers $0 \le k_1(\theta) \le k_2(\theta) \le \ldots \le k_s(\theta) \le N$ such that when considering the $i^{\text{th}}$ selection, where $1 \le i \le s$, we reject the first $k_{i}$ candidates, wait for the $(i-1)^{\text{th}}$ selection,  and then accept the next left-to-right maxima. 
\end{theorem}

\begin{proof}
The algorithms described in Theorem~\ref{winprob-2} (Proposition~\ref{easyexpansion}) produces a strike set which guarantees the optimal winning probability. By Corollary~\ref{cor-2} and $Q_j(N) \ge Q_j^o(N)$, there exists some $0 \le k_i(\theta) \le N-1$ such that $Q_{j}(k) \ge Q^o_{j}(k)$ for $k \ge k_i(\theta)+1$ and $Q_{j}(k) < Q^o_{j}(k)$ for all $k \le k_i(\theta)$, where $1 \le i \le s$. Therefore, an optimal strategy is to reject the first $k_i(\theta)$ candidates and then accept the next left-to-right maxima thereafter. It is also clear that every optimal strategy needs to proceed until the $(i-1)^{\text{th}}$ selection is made before considering the $i^{\text{th}}$ selection. Thus, $k_{i-1}(\theta) \le k_i(\theta)$ for each $i \in \{2, \ldots, s\}$.
\end{proof}

By the definition of the probabilities $Q_j(k), Q^o_j(k), \bar{Q}_j(k)$, we know that they only depend on $\theta, k, N$, and the number of selections left before interviewing the current candidate, i.e., the subscript $j$. Thus, for two different models with $s_1$ and $s_2$ selections respectively (say $s_1 < s_2$), and the same values of $\theta$ and $N$, we have that the thresholds $k'_{s_1 + 1 - j}(\theta)$ for the model with $s_1$ selections and $k''_{s_2 + 1 - j}(\theta)$ for the model with $s_2$ selections are the same whenever $1 \le j \le s_1$. In other words, for each fixed $\theta > 0$, our optimal strategy is right-hand based; and, Corollary~\ref{fixed} holds.

\begin{corollary}\label{fixed}
Let $N$ be a fixed positive integer. For each $\theta>0$, there is a sequence of numbers $a_1(\theta), a_2(\theta), \ldots,$ such that when the number of selections $s \ge 1$ is fixed, then an optimal strategy is the $(a_s(\theta), a_{s-1}(\theta), \ldots, a_1(\theta))$-strategy. In other words, the $(s+1-i)^{\text{th}}$ threshold $k_{s+1-i}(\theta)$ (the $i^{\text{th}}$ from the right) does not depend on the total number of selections allowed (i.e., the value of $s$) and always equals $a_i(\theta)$, for $1\le i \le s$.
\end{corollary}

\begin{example}\label{prefixtree}
To clarify the above observations and concepts, we present an example for the case $\theta = 1$, $s = 2$, and $N = 4$. An optimal strategy is the $(0,1)$-strategy where we accept the first candidate, ask the expert whether this candidate is the best, and then accept the next left-to-right maxima. The optimal winning probability is $17/24$, an improvement of $6/24$ when compared with the optimal winning probability which equals $11/24$ for the case when only one selection is allowed (see Figure~\ref{tree-1}). Note that for each prefix $\sigma \in S_4$, we list the probabilities $Q_1, Q_1^o$ in the first line and the probabilities $Q_2, Q_2^o$ in the second line underneath each prefix shown in Figure~\ref{tree-1}. 
\end{example}

\section{Results for the Mallows Distribution}\label{mallows}
\begin{defn}
Let $P_N (\theta)$ (henceforth $P_N$ to avoid notational clutter) be equal to $1 + \theta + \ldots + \theta^{N-1}$; by convention, we set $P_0 (\theta)= 0$. Furthermore, let $(P_N )!$ be a polynomial in $\theta$ equal to $(P_N)! = P_N P_{N-1} \cdot \ldots \cdot P_1$.
\end{defn}

The following result is well-known and also proved in~\cite{jones2020weighted}.

\begin{lemma} [Lemma 6.2 in~\cite{jones2020weighted},~\cite{mallows1957non}]
We have $$(P_{N})! =  \sum\limits_{\pi \in S_N} \theta^{\#\text{inversions in }\pi}.$$
\end{lemma}

For the set $[1, n+m],$ an ordered $2$-partition of the values into two parts $\Pi_1$ and $\Pi_2$ with $|\Pi_1| = n$ and $|\Pi_2| = m$ is a partition where all values in $\Pi_1$ are positioned before all values in $\Pi_2$, while the internal order within $\Pi_1$ and $\Pi_2$ is irrelevant. 

We define $$B(n,m):= \sum\limits_{\text{All }\Pi_1,\Pi_2 \text{ ordered partitions of }[n+m]} \theta^{\# \text{crossing inversions of }(\Pi_1, \Pi_2)},$$ where a crossing inversion with respect to $(\Pi_1, \Pi_2)$ is an inversions of the form $(a, b)$ where $a \in \Pi_1$, $b \in \Pi_2$, and $a>b$. A straightforward induction argument can be used to prove that if $\theta = 1$ then
$$B(n, m) = {n+m \choose n}.$$

For $n,m \ge 1$, define $${P_{n+m} \choose P_{n}} := \frac{(P_{n+m})!}{(P_{m})! \cdot (P_{n})!}.$$

The following result was established in a paper by the authors of this work~\cite{postdocpaper}.

\begin{lemma}[\cite{postdocpaper}]
For $\theta \neq 1$, $n, m \ge 1$,
$$B(n, m) = \frac{(1-\theta^{n+m}) \cdot (1-\theta^{n+m-1}) \cdot \ldots \cdot (1-\theta^{n+1})}{(1-\theta^m) \cdot (1-\theta^{m-1}) \cdot \ldots \cdot (1-\theta)}, \text{ with } B(n, 0) = B(0, m) = 1.$$
\end{lemma}

Note that 
$$B(n, m) = \frac{P_{n+m} \cdot \ldots \cdot P_{n+2} \cdot P_{n+1}}{P_{m} \cdot \ldots \cdot P_2 \cdot P_1} = \frac{(P_{n+m})!}{(P_{m})! \cdot (P_{n})!} = {P_{n+m} \choose P_{n}}. $$

In order to classify the permutations and compute the winning probabilities according to the selection strategy, we define the following concepts.

\begin{defn}\label{weights-1}
Let $\pi \in S_N$. We say that $\pi$ is {\em $(k_1, \ldots, k_s)$-winnable} if the value $N$ is picked using the $s$-thresholds $(k_1, \ldots, k_s)$-strategy, where $s \ge 1$ is an integer. Furthermore, we define $W(N,k_1, \ldots, k_s)$ to be the sum of all weights of the $(k_1, \ldots, k_s)$-winnable permutations. In other words, $$W(N, k_1, \ldots, k_s):= \sum\limits_{\pi \in S_N \text{ is }(k_1, \ldots, k_s)-\text{winnable}} \theta^{ \text{\# of inversions in }\pi}.$$
\end{defn}

In order to find $W(N, k_1, \ldots, k_s)$ for a given $(k_1, \ldots, k_s)$-strategy, we make use of the following definition.

\begin{defn}\label{weights-2}
We call a permutation $\pi \in S_N$ a $(k_1, \ldots, k_s)$-$\le r$-pickable permutation, for $0 \le r \le s-1$, if the process of applying the $(k_1, \ldots, k_s)$-strategy to $\pi$ uses at most $r$ selections. In addition, we define $T_{\le r}(N, k_1, \ldots, k_{s})$ according to 
$$T_{\le r}(N, k_1, \ldots, k_{s}) = \sum\limits_{(k_1, \ldots, k_s)-\le r\text{-pickable permutations }\pi \in S_{N}} \theta^{\text{\# inversions in }\pi} .$$
\end{defn}

\begin{remark}
For consistency of notation, we allow $N < k_r$ for $1 \le r \le s$. Note that if $N \le k_{r+1}$ then every permutation in $S_N$ uses at most $r$ selections. Moreover, if $N \ge k_{r+1}+1$ then $(k_1, \ldots, k_{r+1})\text{-}\le r$-pickable permutations are equivalent to $(k_1, \ldots, k_{r+2})\text{-}\le r$-pickable permutations, $\ldots$, $(k_1, \ldots, k_{s})\text{-}\le r$-pickable permutations. Thus, we collectively refer to all these permutations as $(k_1, \ldots, k_{r+1})\text{-}\le r$-pickable permutations.
\end{remark}

Another result of interest establishes a formula for $T_0(m,k)$ (i.e., $T_{\le 0}(m,k)$).

\begin{lemma}[\cite{postdocpaper}]\label{relation-t0}
One has $T_0(m,0) = 0$, while for $k \ge 1$, 
$$T_0(m,k)=(\theta^{m-1}+ \ldots+\theta^{m-k}) \cdot (P_{m-1})! = \theta^{m-k} \cdot P_k \cdot (P_{m-1})!.$$
\end{lemma}

We present next the recurrence relation which can be used to find $T_{\le r-1}(m, k_1, \ldots, k_r), \text{ where } 1 \le r \le s$.

\begin{lemma}\label{relation-t}
For each $1 \le r \le s$ and $m \ge k_r+1$, $$T_{\le r-1}(m, k_1, \ldots, k_r) = (P_{m-1})! \cdot \theta^{m-k_r} \cdot P_{k_r} + (P_{m-1})! \cdot \sum\limits_{i = k_r}^{m-1} \frac{T_{\le r-2}(i, k_1, \ldots, k_{r-1})}{(P_{i})!} \cdot \theta^{m-1-i}.$$
\end{lemma}

\begin{proof}
We have to consider two separate cases depending on the position of the value $m$.

\textbf{Case 1:} The value $m$ is at a position $i \in [1, k_r]$. If the value $m$ is at a position $i$ in $[1, k_{r-1}],$ we make at most $r-1$ selections; if the value $m$ is at a position $i \in [k_{r-1}+1, k_r]$, then we either have made at most $r-2$ selections before position $i$, and ended up without any further selections after position $i$; or, we made the $(r-1)^{\text{th}}$ selection at some position $j \in [k_{r-1}+1, i-1]$ (note that using our strategy we cannot make the $r^{\text{th}}$ selection until after position $k_r$), and once again ended up without any further selections after position $i$. In the latter case, we do not select the candidate at position $i$. The other positions can be represented by an arbitrary permutation in $S_{m-1}$. This argument accounts for the term 
$$(\theta^{m-1}+\theta^{m-2}+ \ldots + \theta^{m-k_r}) \cdot (P_{m-1})! = (P_{m-1})! \cdot \theta^{m-k_r} \cdot P_{k_r}.$$

\textbf{Case 2:} The value $m$ is at a position $i \in [k_r+1, m]$. Then the entries at positions $[1, i-1]$ must form a $(k_1, \ldots, k_r)$-$\le (r-2)$ -pickable permutation, since if $r-1$ selections were made before the position $i$, then the $r^{\text{th}}$ selection will occur either before position $i$ or at position $i$. There are no restrictions for entries at positions $[i+1, m]$. The value $m$ itself contributes $\theta^{m-i}$ to the claimed expression. Therefore, for Case 2, we have the following contributing term 
 $$\sum\limits_{i=k_r+1}^{m} \theta^{m-i} \cdot T_{\le r-2}(i-1, k_1, \ldots, k_{r-1})\cdot B(i-1, m-i) \cdot (P_{m-i})! =$$ 
 $$(P_{m-1})! \cdot \sum\limits_{i = k_r+1}^m \frac{T_{\le r-2}(i-1, k_1, \ldots, k_{r-1})}{(P_{i-1})!} \cdot \theta^{m-i}. \qedhere$$ 
\end{proof}

We first address the following special case for which $m = k_r$ and $T_{\le r-1}(m, k_1, \ldots, k_r)$, and use it later to obtain an explicit formula for $W(N, k_1, \ldots, k_s)$.

\begin{lemma}\label{formula-t}
For each $1 \le r \le s$ we have 
$$T_{\le r-1}(k_r, k_1, \ldots, k_r) = (P_{k_r})!,$$ 
since $k_r$ equals the threshold for the $r^{\text{th}}$ selection. Thus, in this case we are not allowed to make the $r^{th}$ selection. For $m \ge k_r+1$,
$$
\frac{T_{\le r-1}(m, k_1, \ldots, k_r)}{(P_{m})!} = \frac{1}{P_m} \cdot \Bigg( \theta^{m-k_{r}} \cdot P_{k_r} + \theta^{m-k_{r-1}-1} \cdot P_{k_{r-1}} \cdot \sum\limits_{i = k_r}^{m-1} \frac{1}{P_i} + \theta^{m-k_{r-2}-2} \cdot P_{k_{r-2}} \cdot \sum\limits_{i_1 = k_r}^{m-1} \frac{1}{P_{i_1}} \sum\limits_{i_2 = k_{r-1}}^{i_1-1} \frac{1}{P_{i_2}} $$
$$ + \ldots + \theta^{m-k_1-r+1} \cdot P_{k_1} \cdot \sum\limits_{i_1 = k_r}^{m-1} \frac{1}{P_{i_1}} \cdot \sum\limits_{i_2 = k_{r-1}}^{i_1-1} \frac{1}{P_{i_2}} \sum \cdots \sum\limits_{i_{r-1} = k_2}^{i_{r-2}-1} \frac{1}{P_{i_{r-1}}} \Bigg).$$
\end{lemma} 

\begin{proof}
The proof is by induction. The case $r = 1$ holds by Lemma~\ref{relation-t0}. Assume the result is true for all number of queries less than $r$, where $1 \le r \le s$. We prove that it holds for $r$ as well. By Lemma~\ref{relation-t}, $T_{\le r-1}(m, k_1, \ldots, k_r)$ can be expressed in terms of $T_{\le r-2}(i, k_1, \ldots, k_{r-1})$, $k_r \le r \le m-1$. Thus, by using the formula for $T_{\le r-2}(i, k_1, \ldots, k_{r-1})$, guaranteed by the inductive hypothesis, we obtain the claimed formula for $\frac{T_{\le r-1}(m, k_1, \ldots, k_r)}{(P_m)!}$. The actual derivations are omitted.
\end{proof}


To find $W(N,k_1, \ldots, k_s)$, we use a special-case result of Jones~\cite{jones2020weighted} for $s=1$.

\begin{theorem}[Jones~\cite{jones2020weighted}]\label{jonesmain}
For $k_1 \ge 1$,
$$W(N,k_1) = \theta^{N-k_1-1} \cdot (P_{N-1})! \cdot P_{k_1} \cdot \sum\limits_{i = k_1}^{N-1} \frac{1}{P_i}.$$
\end{theorem}

\begin{lemma}\label{relations-s}
For $1 \le r \le s$ and $N \ge k_r+1$, 
$$W(N, k_1, \ldots, k_r) = \theta \cdot P_{N-1} \cdot W(N-1, k_1, \ldots, k_r) + T_{\le r-1}(N-1, k_1, \ldots, k_r).$$
\end{lemma}

\begin{proof}
We consider the value at position $N$. There are two possible cases to consider.

\textbf{Case 1:} The value is $N$. Then, in order to select the value $N$ at the last position, we require that there are at most $r-1$ selections made for the first $N-1$ positions. This gives rise to the term $T_{\le r-1}(N-1, k_1, \ldots, k_r)$. 

\textbf{Case 2:} The value is $i \in \{1, \ldots, N-1\}$. Then, the first $N-1$ positions form a $(k_1, \ldots, k_r)$-winnable permutation. The value $i$ at the last position contributes $\theta^{N-i}$ to the overall expression. Therefore, the total contribution equals 
$$(\theta^{N-1} + \ldots + \theta) \cdot W(N-1, k_1, \ldots, k_r) = \theta \cdot P_{N-1} \cdot W(N-1, k_1, \ldots, k_r). \qedhere$$ 
\end{proof}

Lemma~\ref{formula-s} describes the winning probability for a given $(k_1, \ldots, k_s)$-strategy, with $0 \le k_1 \le \ldots \le k_s \le N$.

\begin{lemma}\label{formula-s}
For each $1 \le s$ and $N \ge k_s+1$ we have
$$\frac{W(N, k_1, \ldots, k_s)}{(P_N)!} = \frac{1}{P_N} \cdot \Bigg( \Big(\theta^{N-k_1-1} \cdot P_{k_1} \cdot \sum\limits_{i = k_1}^{k_2-1} \frac{1}{P_i}  \Big)  +  \Big(\theta^{N-k_2-1} \cdot P_{k_2} \cdot \sum\limits_{i = k_2}^{k_3-1} \frac{1}{P_i} $$
$$+ \theta^{N-k_1-2} \cdot P_{k_1} \cdot \delta(k_2, k_3) \cdot \sum\limits_{i_1 = k_2+1}^{k_3-1} \frac{1}{P_{i_1}} \sum\limits_{i_2=k_2}^{i_1-1} \frac{1}{P_{i_2}}  \Big) +  \Big( \theta^{N-k_3-1} \cdot P_{k_3} \cdot \sum\limits_{i=k_3}^{k_4-1} \frac{1}{P_i} $$
$$+ \delta(k_3, k_4) \cdot (\theta^{N-k_2-2} \cdot P_{k_2} \cdot \sum\limits_{i_1 = k_3+1}^{k_4-1} \frac{1}{P_{i_1}} \sum\limits_{i_2 = k_3}^{i_1-1} \frac{1}{P_{i_2}} + \theta^{N-k_1-3} \cdot P_{k_1} \cdot \sum\limits_{i_1 = k_3+1}^{k_4-1} \frac{1}{P_{i_1}} \sum\limits_{i_2 = k_3}^{i_1-1} \frac{1}{P_{i_2}} \sum\limits_{i_3 = k_2}^{i_2-1} \frac{1}{P_{i_3}})  \Big)
$$
$$ + \ldots +  \Big( \theta^{N-k_s-1} \cdot P_{k_s} \cdot \sum\limits_{i = k_s}^{N-1} \frac{1}{P_i} + $$
$$
\delta(k_s, N) \cdot (\theta^{N-k_{s-1}-2} \cdot P_{k_{s-1}} \cdot \sum\limits_{i_1 = k_s+1}^{N-1} \frac{1}{P_{i_1}} \sum\limits_{i_2 = k_s}^{i_1-1} \frac{1}{P_{i_2}} + \ldots + \theta^{N-k_1-s} \cdot P_{k_1} \cdot \sum\limits_{i_1 = k_s+1}^{N-1} \frac{1}{P_{i_1}} \sum\limits_{i_2 = k_s}^{i_1-1} \frac{1}{P_{i_2}} \cdots \sum\limits_{i_s = k_2}^{i_{s-1}-1} \frac{1}{P_{i_s}})  \Big) \Bigg),$$
where $\delta(k_i, k_{i+1})$ (by default, $k_{s+1} = N$) equals $1$ if $k_{i+1} \ge k_{i} + 2$ and $0$ otherwise, for $2 \le i \le s$.
\end{lemma}

\begin{proof}
We prove the claim by induction. The case $s = 1$ holds by Theorem~\ref{jonesmain}. Assume the claim holds for $W(N, k_1, \ldots, k_r)/(P_N)!,$ where $r < s$. We prove the result holds for $r = s$. By Lemma~\ref{relations-s},

$$W(N, k_1, \ldots, k_s) = \theta \cdot P_{N-1} \cdot W(N-1, k_1, \ldots, k_s) + T_{\le s-1}(N-1, k_1, \ldots, k_s)$$
$$= \theta \cdot P_{N-1} \cdot (\theta \cdot P_{N-2} \cdot W(N-2, k_1, \ldots, k_s) + T_{\le s-1}(N-2, k_1, \ldots, k_s)) + T_{\le s-1}(N-1, k_1, \ldots, k_s) = \ldots = $$
$$\theta^{N-k_s} \cdot \frac{(P_{N-1})!}{(P_{k_s-1})!}\cdot W(k_s, k_1, \ldots, k_{s-1}, k_s) + \sum\limits_{i = 1}^{N-k_s} \theta^{i-1} \cdot \frac{(P_{N-1})!}{(P_{N-i})!} \cdot T_{\le s-1}(N-i, k_1, \ldots, k_s)=$$
$$\theta^{N-k_s} \cdot \frac{(P_{N-1})!}{(P_{k_s-1})!}\cdot W(k_s, k_1, \ldots, k_{s-1}) + \sum\limits_{i = 1}^{N-k_s} \theta^{i-1} \cdot \frac{(P_{N-1})!}{(P_{N-i})!} \cdot T_{\le s-1}(N-i, k_1, \ldots, k_s).$$

By the inductive hypothesis, we can use the formula for $W(k_s, k_1, \ldots, k_{s-1})$, and the formula of \\$T_{\le s-1}(m, k_1, \ldots, k_s)$ in Lemma~\ref{formula-t} to establish the claim. The derivations are omitted.
\end{proof}

\begin{remark} A $(k_1, \ldots, k_s)$-strategy with $k_{i} = k_{i+1}$ is equivalent to a $(k_1', \ldots, k_s')$-strategy with $k_j' = k_j$ for $j \neq i+1$ and $k'_{i+1} = k_i + 1$, where $1 \le i \le s$. Thus, we may assume that when $k_1 = 0$ we have $k_2 \ge 1$. Furthermore, for $k_1 = 0$,
$$W(N, 0, k_2, \ldots, k_s) = \theta^{N-1} \cdot (P_{N-1})! + W(N, k_2, \ldots, k_{s}) \text{ and } \frac{W(N, k_1, \ldots, k_s)}{(P_N)!} = \frac{\theta^{N-1}}{P_N} + \frac{W(N, k_2, \ldots, k_s)}{(P_N)!}.$$
\end{remark}

We henceforth focus on the optimal strategy when $N \to \infty$. Let $0< \theta < 1$ be fixed. The optimal strategy is right-hand based, i.e., we wait until the very end to make even the first selection (the difference between $N$ and $k_1$ is a fixed number for each given $\theta,$ with $0<\theta<1$).

\begin{theorem}\label{asm-s}
The asymptotically optimal $(k_1, \ldots, k_s)$-strategy for $\theta < 1$ has to satisfy $N-k_1 \not \to \infty$.
\end{theorem}
\begin{proof}
The proof is postponed to the Appendix (Section~\ref{appendix}).
\end{proof}

Let $\theta > 1$ be fixed. The optimal strategy is left-hand based, i.e., the threshold for the $s^{\text{th}}$ selection is a constant which only depends on $\theta$.

\begin{theorem}\label{asm-l}
The asymptotically optimal $(k_1, \ldots, k_s)$-strategy for $\theta > 1$ has to satisfy $k_s \not \to \infty$.
\end{theorem}
\begin{proof}
The proof is postponed to the Appendix (Section~\ref{appendix}).
\end{proof}

By Lemma~\ref{asm-s}, we know that for $N \to \infty$, $0 < \theta < 1$, the optimal strategy is a $(k_1, \ldots, k_s)$-strategy for some $k_1 \le \ldots \le k_s $ with $N-k_1 \not \to \infty$. By Lemma~\ref{formula-s} and the fact that $\frac{1}{P_i} = \frac{1-\theta}{1 - \theta^i} \to 1 - \theta$ when $i \to \infty$ and $0 < \theta < 1$, the optimal probability is a function of $N-k_1, \ldots, N-k_s$ and does not depend on $N$. By Corollary~\ref{fixed}, we also know that for a fixed $0< \theta < 1$, there is an optimal strategy satisfying $k_i = a_{s+1-i}(\theta)$, where $1 \le i \le s$. To simplify notation, we henceforth use $a_j$ to denote $a_j(\theta)$,  $1 \le j \le s$.

First, note that $\frac{1}{1 - \theta^i} \to 1$ when $0 < \theta < 1$ and $i \to \infty$. Let 
$$H'_1 =  \theta^{N-a_1} \cdot (\frac{1}{\theta}-1) \cdot \sum\limits_{i = a_1}^{N-1} 1,$$
$$H'_2 = \theta^{N-a_2} \cdot  \Big((\frac{1}{\theta}-1) \cdot \sum\limits_{i = a_2}^{a_1-1} 1+ (\frac{1}{\theta}-1)^2 \cdot \delta(a_1, N) \cdot \sum\limits_{i_1 = a_1+1}^{N-1} \sum\limits_{i_2=a_1}^{i_1-1} 1  \Big),$$
$$\hspace{-8mm}H'_3 = \theta^{N-a_3} \cdot  \Big( (\frac{1}{\theta}-1) \cdot \sum\limits_{i = a_2}^{a_3-1} 1 + (\frac{1}{\theta}-1)^2 \cdot \delta(a_2, a_1) \cdot \sum\limits_{i_1 = a_2+1}^{a_1-1} \sum\limits_{i_2=a_2}^{i_1-1} 1 + (\frac{1}{\theta}-1)^3 \cdot \delta(a_1, N) \cdot \sum\limits_{i_1 = a_1+1}^{N-1} \sum\limits_{i_2 = a_1}^{i_1-1} \sum\limits_{i_3 = a_2}^{i_2-1} 1  \Big), $$
$$\ldots$$
$$\hspace{-8mm} H'_s =\theta^{N-a_s} \cdot  \Big( (\frac{1}{\theta}-1) \cdot \sum\limits_{i = a_s}^{a_{s-1}-1} 1 +  (\frac{1}{\theta}-1)^2 \cdot \delta(a_{s-1}, a_{s-2}) \cdot \sum\limits_{i_1 = a_{s-1}+1}^{a_{s-2}-1}  \sum\limits_{i_2 = a_{s-1}}^{i_1-1} 1 $$
$$+ (\frac{1}{\theta}-1)^3 \cdot \delta(a_{s-2}, a_{s-3}) \cdot \sum\limits_{i_1 = a_{s-2}+1}^{a_{s-3}-1} \sum\limits_{i_2 = a_{s-2}}^{i_1-1} \sum\limits_{i_3 = a_{s-1}}^{i_2-1} 1 + \ldots + (\frac{1}{\theta}-1)^{s} \cdot \delta(a_1, N)  \cdot \sum\limits_{i_1 = a_1+1}^{N-1} \sum\limits_{i_2 = a_1}^{i_1-1} \sum\limits_{i_3 = a_2}^{i_2-1}  \cdots \sum\limits_{i_s = a_{s-1}}^{i_{s-1}-1} 1  \Big),$$ 
where $\delta(x,y) = 1$ if $y \ge x+2$, and equals $0$ otherwise.

Hence, after reorganizing and simplifying the formula presented in Lemma~\ref{formula-s}, we have 
$$\lim\limits_{N \to \infty} \frac{W(N, a_s, \ldots, a_1)}{(P_N)!} = \lim\limits_{N \to \infty} \sum\limits_{i = 1}^{s} H'_i =: P'.$$

By Lemma~\ref{asm-l}, we know that for $N \to \infty$, $\theta > 1$, the optimal strategy is a $(k_1, \ldots, k_s)$-strategy for some $0 \le k_1 \le \ldots \le k_s \not \to \infty$. By Lemma~\ref{formula-s}, the fact that $\frac{1}{P_i} = \frac{\theta - 1}{\theta^i - 1}$ when $ \theta > 1$, and the observation that each of the sums $$\sum\limits_{i_1 = k_s+1}^{N-1} \frac{1}{\theta^{i_1}-1} \sum\limits_{i_2 = k_s}^{i_1-1}\frac{1}{\theta^{i_2}-1} \cdots \sum\limits_{i_r = k_{s+2-r}}^{i_{r-1} - 1} \frac{1}{\theta^{i_r}-1}, \quad \text{ where } 2 \le r \le s,$$ converges when $N \to \infty$, it follows that the optimal probability is a function of $k_1, \ldots, k_s$ and does not depend on $N$. 

By Corollary~\ref{fixed}, we similarly know that for a fixed $\theta > 1$, there is an optimal strategy that satisfies $k_i = a_{s+1-i}(\theta)$, where $1 \le i \le s$. 
Next, let 
$$H''_1 =  (1 - \frac{1}{\theta^{a_1}}) \cdot (1-\frac{1}{\theta}) \cdot \sum\limits_{i = a_1}^{N-1} \frac{1}{\theta^i - 1},$$
$$H''_2 = (1 - \frac{1}{\theta^{a_2}}) \cdot  \Big( (1-\frac{1}{\theta}) \cdot \sum\limits_{i = a_2}^{a_1-1} \frac{1}{\theta^i - 1} + (1 - \frac{1}{\theta})^2 \cdot \delta(a_1, N) \cdot \sum\limits_{i_1 = a_1+1}^{N-1} \frac{1}{\theta^{i_1} - 1} \sum\limits_{i_2 = a_1}^{i_1-1} \frac{1}{\theta^{i_2} - 1}  \Big),$$
$$H''_3 = (1 - \frac{1}{\theta^{a_3}}) \cdot  \Big( (1-\frac{1}{\theta}) \cdot \sum\limits_{i = a_3}^{a_2-1} \frac{1}{\theta^i - 1} + (1-\frac{1}{\theta})^2 \cdot \delta(a_2, a_1) \cdot \sum\limits_{i_1 = a_2+1}^{a_1-1} \frac{1}{\theta^{i_1} - 1} \sum\limits_{i_2= a_2}^{i_1-1} \frac{1}{\theta^{i_2}-1} + $$
$$(1-\frac{1}{\theta})^3 \cdot \delta(a_1, N) \cdot \sum\limits_{i_1 = a_1+1}^{N-1} \frac{1}{\theta^{i_1} - 1} \sum\limits_{i_2 = a_1}^{i_1-1} \frac{1}{\theta^{i_2}-1} \sum\limits_{i_3 = a_2}^{i_2-1} \frac{1}{\theta^{i_3}-1}  \Big), $$
$$\ldots$$
$$H''_s =(1 - \frac{1}{\theta^{a_s}}) \cdot  \Big( (1 - \frac{1}{\theta}) \cdot \sum\limits_{i = a_s}^{a_{s-1}-1} \frac{1}{\theta^i - 1} +  (1 - \frac{1}{\theta})^2 \cdot \delta(a_{s-1}, a_{s-2}) \cdot \sum\limits_{i_1 = a_{s-1}+1}^{a_{s-2}-1} \frac{1}{\theta^{i_1}-1} \sum\limits_{i_2 = a_{s-1}}^{i_1-1} \frac{1}{\theta^{i_2}-1} +$$
$$(1 - \frac{1}{\theta})^3 \cdot \delta(a_{s-2}, a_{s-3}) \cdot \sum\limits_{i_1 = a_{s-2}+1}^{a_{s-3}-1} \frac{1}{\theta^{i_3}-1} \sum\limits_{i_2 = a_{s-2}}^{i_1-1} \frac{1}{\theta^{i_2}-1} \sum\limits_{i_3 = a_{s-1}}^{i_2-1} \frac{1}{\theta^{i_3}-1} + \ldots + $$
$$ (1 - \frac{1}{\theta})^{s} \cdot \delta(a_{1}, a_{N}) \cdot \sum\limits_{i_1 = a_1+1}^{N-1} \frac{1}{\theta^{i_1}-1} \sum\limits_{i_2 = a_1}^{i_1-1} \frac{1}{\theta^{i_2}-1} \sum\limits_{i_3 = a_2}^{i_2-1} \frac{1}{\theta^{i_3}-1} \cdots \sum\limits_{i_s = a_{s-1}}^{i_{s-1}-1} \frac{1}{\theta^{i_s}-1}  \Big),$$ where $\delta(x,y) = 1$ if $y \ge x+2$, and equals $0$ otherwise.

After reorganizing terms and simplifying the formula presented in Lemma~\ref{formula-s} we obtain 
$$\lim\limits_{N \to \infty} \frac{W(N, a_s, \ldots, a_1)}{(P_N)!} = \lim\limits_{N \to \infty} \sum\limits_{i = 1}^{s} H''_i =: P''.$$

By Theorems~\ref{asm-s} and~\ref{asm-l}, and the formulas above, when $0< \theta < 1$, $P'$ is a function of $b_1 := N-a_1, \ldots, b_s:=N-a_s$ with $1 \le b_1 < b_2 < \ldots < b_s,$ and is maximized at some $b'_1, \ldots b'_s$ with $b'_s \not \to \infty$. When $\theta > 1$, $P''$ is a function of $a_1, \ldots, a_s$ with $0 \le a_s \le a_{s-1} \le \ldots \le a_1,$ and is maximized at some $a''_1, \ldots a''_s$ with $a''_1 \not \to \infty$. By Corollary~\ref{fixed}, for each fixed $0<\theta<1$ we can pick some large enough constant, say $1000$ ($1000$ is a large enough value to serve as a ``proxy'' for $\infty$, as can be seen from our subsequent numerical computations for values of $\theta$ as small as $0.01$ and $s$ as large as $5$), and  perform a computer search for $a''_1$ (and $b'_1$). Recall that this value characterizes our optimal strategy for $s = 1$. We then proceed to search for $a''_2 \le a''_1$ ($b'_2 \ge b'_1$ ), which jointly with $a''_1$ ($b'_1$) characterize our optimal strategy for $s = 2$. We then search for $a''_3,\ldots$ ($b'_3,\ldots$) following the same procedure.

\section{Numerical Results}\label{numerical}

All results presented herein hold for $N \to \infty$. By Corollary~\ref{fixed}, an optimal strategy for $s' = s-1 \ge 0$ queries and a fixed $\theta > 1$ is an $(a_s(\theta), \ldots, a_1(\theta))$-strategy, where $a_i(\theta) \not \to \infty$ for $1 \le i \le s$. For the case $0 < \theta < 1$, we let $b_i(\theta) = N - a_i(\theta)$, where $i \ge 1$. An optimal strategy for $s' = s - 1 \ge 0$ queries and a fixed $0 < \theta < 1$ is an $(N-b_s(\theta), \ldots, N-b_1(\theta))$-strategy, where $b_i(\theta) \not \to \infty$ for each $1 \le i \le s$. Since the maximum probability of success increases very slowly as $s$ increases above $5$, and since the corresponding computational times increase as well, we only numerically computed the maximum probability of success and an optimal strategy for $s \leq 5$. The results are presented in Table~\ref{numericalresult-1}. 

As expected, for each fixed $s \ge 1,$ the smallest probability of success arises for $\theta = 1$. Moreover, if $s$ is a fixed positive integer, then the optimal probability of success tends to $1$ when $\theta \to 0$ as well as when $\theta \to \infty$. It is also intuitively clear that when $0 < \theta < 1$, as $\theta$ decreases, the Mallows distribution concentrates around the identity permutation $[12 \cdots N]$; in this setting, a $(N-b_s, \ldots, N-b_1)$-strategy with ``small'' values of $b_1, \ldots, b_s$ has a high probability to identify the best candidate. When $\theta > 1$, the Mallows distribution concentrates around the permutation $[N (N-1) \cdots 21]$ and as $\theta$ increases, a $(a_s, \ldots, a_1)$-strategy with ``small'' values of $a_1, \ldots, a_s$ has a high probability to identify the best candidate. Note that for $\theta > 1$ and a fixed $(a_s, \ldots, a_1)$-strategy, the value of $\theta$ for which the probability of success is maximized does not occur when $\theta \to 1+$ or $\theta \to \infty$ but for some other fixed value. This is the reason for the observable small decreases in the optimal probability of success for $\theta>1$, depicted in Figure~\ref{best-1}. 


For each fixed $i \ge 1$, $a_i \to \infty$ as $\theta \to 1+$ and the number $b_i \to \infty$ as $\theta \to 1-$. 
 Furthermore, for each fixed $\theta > 0$, the optimal probability of winning increases as $s$ increases and tends to $1$ as $s \to \infty$. This is also intuitively clear, since for fixed $\theta > 0$ there is a better chance of succeeding when more queries are allowed, and we are guaranteed to succeed if we have infinitely many selections. Note that this probability increases dramatically when $s$ decreases. In particular, for $s > 5$ the smallest probability of success exceeds $0.9$; this is the main reason why we focus our attention on results for $s \leq 5$.

\vspace{3mm}

\begin{table}
\hspace{-8mm}
\begin{center}
\begin{tabular}{|c|c|c|c|c|c|c|c|c|c|c|}
\hline
$\theta$ & $b_1$ & $p$ & $b_2$  & $p$ & $b_3$  & $p$ & $b_4$  & $p$ & $b_5$  & $p$ \\ \hline

0.01& 1 & 0.99 & 2 & 0.9999 & 3 & 0.999999 & 4  & 0.99999999 & 5  & 0.9999999999 \\ \hline
0.1& 1 & 0.9 & 2 & 0.99 & 3 & 0.999 & 4  & 0.9999 & 5 & 0.99999          \\ \hline
0.2& 1 & 0.8 & 2 & 0.96 & 3 & 0.992& 4   & 0.9984 & 5   & 0.99968  \\ \hline
0.3& 1 & 0.7 & 2 & 0.91& 3 & 0.973  & 4  & 0.9919  & 5   & 0.99757            \\ \hline
0.4& 1 & 0.6 & 2 & 0.84 & 3 & 0.936  & 4   & 0.9744  & 5    & 0.98976      \\ \hline
0.5& 1 & 0.5 & 2 & 0.75 & 3 & 0.875 & 4   & 0.9375 & 5     & 0.96875    \\ \hline
0.6& 2 & 0.48 & 3 & 0.72 & 5 & 0.84672 & 6  & 0.916992 &7   & 0.955008      \\ \hline
0.7& 3 & 0.441 & 5 & 0.67767 & 6 & 0.814527  & 8  & 0.89181519  &9    & 0.9367475     \\ \hline
0.8& 4 & 0.4096 & 7 & 0.6455296 & 9 & 0.78394163  & 12 & 0.86742506 & 14 & 0.91836337   \\ \hline
0.9& 9 & 0.38742049 & 14 & 0.61618841 & 19 & 0.75683265  & 24 & 0.84462315 & 28 & 0.90023365   \\ \hline
0.91& 11 & 0.38552196 & 16 & 0.61384283 & 21 & 0.75431993 & 26  & 0.84249939   & 31    & 0.8984815 \\ \hline
0.92& 12 & 0.38365188  & 18 & 0.61122396 & 24 & 0.75183545 & 30  & 0.84029209  & 35 &  0.89669505 \\ \hline
0.93& 14 & 0.38150867 & 21 & 0.60859444 & 28 & 0.74920064 &  34  & 0.83810534 & 40   & 0.89490443 \\ \hline
0.94& 16 & 0.37948013 & 25 & 0.60588389 & 32 & 0.74670942 &  40  & 0.83589721   & 47   & 0.89310511 \\ \hline
0.95& 19 & 0.3773536  & 29 & 0.60328914 & 39 & 0.74418598 &  48  & 0.8337195   & 56 & 0.89132711\\ \hline
0.96& 24 & 0.37541325 & 37 & 0.60083222 & 49 & 0.74172818 &  60  & 0.83157908 & 70 &  0.88956127  \\ \hline
0.97& 33 & 0.37353448 & 50 & 0.59832096 & 65 & 0.73930464 &  80  & 0.82945187 & 94 & 0.88780406\\ \hline
0.98& 49 & 0.37160171 & 74 & 0.59585024 & 97 & 0.73687357 &  119 & 0.82732738  & 141 & 0.88604612\\ \hline
0.99& 99 & 0.36972964 & 149 & 0.59341831 & 195 & 0.73448001 & 239 & 0.82521971 & 282 & 0.8842961 \\ \hline
$\theta$ & $a_1$ & $p$  & $a_2$  & $p$ & $a_3$ & $p$  & $a_4$ & $p$  & $a_5$  & $p$ \\ \hline
1.01&46 & 0.36918367 & 25   & 0.59372585   & 15 &  0.73609875 &  9 &  0.82818603& 6       & 0.8884655 \\ \hline
1.02&23 & 0.37052858  & 12  &  0.59643023 & 7 &  0.74010572  &   4 & 0.83314574  &3  & 0.89440668\\ \hline
1.03&15 &   0.37184338  & 8 &    0.59927181 & 5  &  0.74430003  &  3& 0.83883113 &2 &0.90046661 \\ \hline
1.04&11 &    0.37307045   & 6 &  0.60209564 &  4 &   0.74764429 & 2 &  0.84437558 &1&0.90873876 \\ \hline
1.05&9  &     0.37453849  & 5 &  0.60494232 & 3 &  0.75260449  & 2 &   0.84770309&1&0.9146162 \\ \hline
1.06&8  &     0.37555657   & 4 & 0.60780968 & 2 & 0.75709372 &  1 &  0.85599814&1&0.91585315 \\ \hline
1.07&6  &     0.376652     & 3 & 0.60956208& 2 &  0.76137673 & 1& 0.86299515 &0& 0.92841571\\ \hline
1.08&6  &   0.3782214   &3&  0.61365158& 2 & 0.76357932 &1 &  0.86737476 &0&0.94144883  \\ \hline
1.09&5  &  0.37998224 &3 &  0.61493159 & 1 & 0.7678056  & 1&  0.86916554&0& 0.95173435\\ \hline
1.1& 5  &  0.38013275  &2 &   0.61811891 & 1  & 0.77490222  &1& 0.86897463&0& 0.95988372\\ \hline
1.2& 2  &    0.3946616   & 1 &  0.65166097 & 0 &  0.81832763 &0& 0.95363085 &0& 0.99176211\\ \hline
1.3& 1  &     0.40196949 & 1 &  0.66305426&  0 & 0.89382349 & 0 & 0.98072423&0&0.99771086\\ \hline
1.4& 1 &    0.42452167  & 0 &  0.71023596 & 0  &  0.93366818 &0&0.99101943&0& 0.9992478\\ \hline
1.5& 1 &   0.43301723   &0& 0.76635056& 0 &0.95655232 &0& 0.99547373&0& 0.99972306 \\ \hline
1.6& 1 &   0.43330022 & 0&   0.80830022 &  0& 0.97048698 & 0 & 0.99757918&0&0.99988899\\ \hline
1.7& 1 &   0.42868095 & 0&    0.84044565& 0 & 0.97935572 & 0& 0.99864254&0&0.9999524  \\ \hline
1.8& 0 &   0.44444444  & 0&  0.86557747& 0 &  0.98520273  &  0& 0.9992085&0& 0.99997843 \\ \hline
1.9& 0 &    0.47368421 & 0 &  0.88555846 & 0 &  0.98917129  & 0 & 0.999523&0&0.99998975 \\ \hline
2& 0  & 0.5   & 0 & 0.90167379 & 0 & 0.99193195 & 0 &0.99970422 & 0 & 0.99999493              \\ \hline
3& 0  &  0.66666667  &0&0.969846&0 & 0.99918728 &0& 0.99999306&0& 0.99999998          \\ \hline
4& 0  &  0.75&0&   0.98686745 &0& 0.99983997&0& 0.99999953 &0& 0.9999999999           \\ \hline
5& 0 & 0.8 &0& 0.99310967 &0& 0.99995498&0& 0.99999994&0& 0.9999999999      \\ \hline    
\end{tabular}
\caption{Maximum probability of success $p$ and an optimal strategy for $0<\theta<1$ and $\theta>1$ under the Mallows model. }\label{numericalresult-1}
\end{center}
\end{table}

\begin{figure}
\begin{center}
  \includegraphics[scale=0.6]{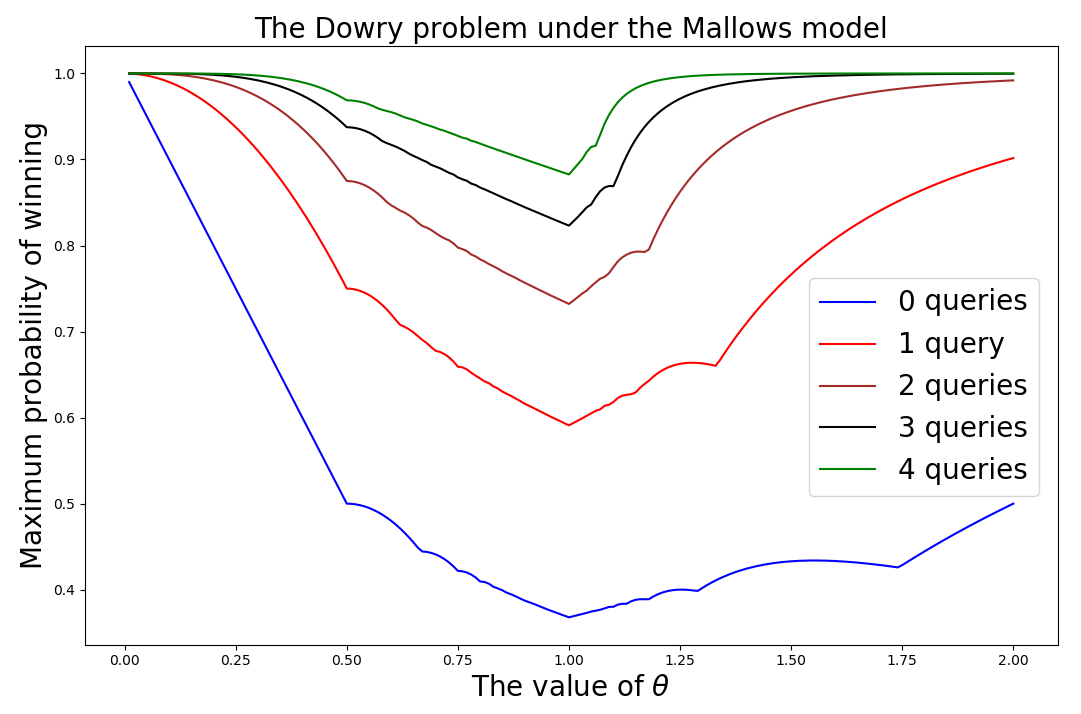}
  \caption{The maximum probability of success using our optimal strategy under the Mallows model. }\label{best-1}
\end{center}
\vspace{-8mm}
\end{figure}

\vspace{-5mm}
\section{Expected Number of Queries and Interviewed Candidates}\label{expect}

The maximum probability of winning for the Dowry model with $s$ selections and the query-based model with $s-1$ queries are the same, as both models have a budget of $s$ selections and the goal is to choose the best candidate (note that this claim holds for all Mallows parameters, but that the Dowry problem has -- until this work -- only been studied for a uniform distribution of candidate orders for which $\theta=1$). However, the expected stopping times are very different. Under the query-based model, the process immediately terminates after obtaining a positive answer from the expert. On the other hand, the decision making entity continues to interview the remaining candidates after a selection is made (provided there is a selection left) under the Dowry setting, as it has no information about weather the current candidate is the best. 

Furthermore, since a query to an expert is costly in practice, it is also of interest to examine the expected number of queries or interviews made during the process.

\subsection{Expected Number of Queries}

In this setting, we are interested in two expectations: \textbf{Unconditional expectations:} The expected number of selections made using the optimal strategy (described in Table~\ref{numericalresult-1}); \textbf{Conditional expectations:} The expected number of selections made conditioned on the event that the best candidate is selected using the optimal strategy (described in Table~\ref{numericalresult-1}).

\begin{claim}
The conditional and unconditional expected number of selections for the query-based model is the same as for the Dowry model.
\end{claim}

\begin{proof}
Assume that under the query-based model we made a total $r$ selections, where $0 \le r \le s$, using our optimal $(k_1, \ldots, k_s)$-strategy with $k_i = a_{s+1-i}(\theta),$ $1 \le i \le s$. We show next that the Dowry model makes the same number of selections.

First, assume that $r = 0$. Then the value $N$ is located at a position $j \le k_1$, since otherwise there is at least one left-to-right maxima at a position in $[k_1, N]$ and thus our optimal query strategy will result in at least one selection in both models.

Next, assume that $1 \le r \le s-1$ and that the $r^{\text{th}}$ query is made at a position $j \ge k_r+1$. The value $N$ must be at a position $h$ such that $h \ge j$, since if $h < j$ there would have been no query at position $j$. 

Under the above assumption, we proceed as follows. First, assume that $h = j$. In the Dowry model, the $1^{\text{st}}$, $\ldots$, $(r-1)^{\text{th}}$ selections/queries are the same as the expert keeps giving a negative answer. We then pick up the value $N$ at position $j$ (without knowing this fact in the Dowry model) and examine the list until the end; we do not make the $(r+1)^{\text{th}}$ selection since there is no left-to-right maxima after position $j$. Second, assume that $h > j$. Then we must have $h \le k_{r+1}$ since otherwise the $(r+1)^{\text{th}}$ selection will be made. In the Dowry model, the $1^{\text{st}}$, $\ldots$, $(r-1)^{\text{th}}$, and $r^{\text{th}}$ selections are the same as those in the query-based model as the latter gives negative answers. We do not make another selection until position $k_{r+1}$ since the $(r+1)^{\text{th}}$ selection is not allowed until after position $k_{r+1}$, and we cannot perform the $(r+1)^{\text{th}}$ selection after position $k_{r+1}$ since there is no left-to-right maxima following position $h$. 

The last case to consider is $r = s$. We used all $s-1$ queries to query an expert, received $s-1$ negative answers, and then made a final decision at position $j \ge k_s+1$. Under the Dowry model, we made the same $s-1$ selections, without knowing that they are not the best, then made the final selection at the same position $j$.

Similarly, we can show that if $r$, $1 \le r \le s$, selections are made in the Dowry model then the same number of selections will be made using the query-based model. Similar arguments apply for the case when we condition on the event of identifying the best candidate.   
\end{proof}

\begin{defn}\label{weights-3}
We call a permutation $\pi \in S_N$ {\em exactly $r$-$(k_1, \ldots, k_s)$-winnable} if the best value $N$ is selected as the $r^{\text{th}}$ selection when using the $(k_1, \ldots, k_s)$-strategy, where $1 \le r \le s$. For simplicity, we abbreviate the name to $r$-winnable whenever the strategy is clear. Similarly to Definition~\ref{weights-1}, we define for $1 \le r \le s$ that
$$W_r(N, k_1, \ldots, k_s) = \sum\limits_{r\text{-winnable }\pi \in S_N} \theta^{c(\pi)}.$$
\end{defn}

\begin{defn}\label{weights-4}
We call a permutation $\pi \in S_N$ {\em exactly $r$-$(k_1, \ldots, k_s)$-pickable} if it results in exactly $r$ selections using the $(k_1, \ldots, k_s)$-strategy, where $0 \le r \le s$. We also abbreviate this reference to $r$-pickable whenever the strategy is clear. Similarly to Definition~\ref{weights-2}, we define for $0 \le r \le s$ that
$$T_r(m, k_1, \ldots, k_s) = \sum\limits_{r\text{-pickable }\pi \in S_N} \theta^{c(\pi)}.$$
\end{defn}
 
By Definition~\ref{weights-1} and~\ref{weights-3}, the following result is straightforward.

\begin{proposition}
One has 
$W(N, k_1, \ldots, k_s) = \sum\limits_{i = 1}^{s} W_i(k_1, \ldots, k_s).$
\end{proposition}

\subsubsection{The Unconditional Expectations}

By Definition~\ref{weights-4}, it is straightforward to see that the expected number of selections made using the $(k_1, \ldots, k_s)$-strategy equals 

$$\sum\limits_{i = 1}^{s} i \cdot \frac{T_i(N, k_1, \ldots, k_s)}{(P_N)!}.$$

The following Lemma is a consequence of Definitions~\ref{weights-2} and~\ref{weights-4}.

\begin{lemma}
For $N \ge k_s+1$, we have $T_0(N, k_1, \ldots, k_s) = T_0(N, k_1)$, and for $1 \le i \le s$, $$T_i(N, k_1, \ldots, k_s) = T_{\le i}(N, k_1, \ldots, k_s) - T_{\le i-1}(N, k_1, \ldots, k_{s}).$$
\end{lemma}

Every permutation uses at most $s$ selections by the definition of the $(k_1, \ldots, k_s)$-strategy. By Lemma~\ref{formula-t}, when $\theta > 1$ and $1 \le r \le s$, we have

$$\lim\limits_{N \to \infty} \frac{T_{\le r-1}(N, k_1, \ldots, k_s)}{(P_N)!} = \lim\limits_{N \to \infty} \frac{T_{\le r-1}(N, k_1, \ldots, k_r)}{(P_N)!} = (1-\frac{1}{\theta^{k_r}}) + $$
$$(1 - \frac{1}{\theta}) \cdot (1 - \frac{1}{\theta^{k_{r-1}}}) \cdot \sum\limits_{i = k_r}^{\infty} \frac{1}{\theta^i - 1}+(1 - \frac{1}{\theta})^2 \cdot (1 - \frac{1}{\theta^{k_{r-2}}}) \cdot \sum\limits_{i_1 = k_r}^{\infty} \frac{1}{\theta^{i_1} - 1} \sum\limits_{i_2 = k_{r-1}}^{i_1 - 1} \frac{1}{\theta^{i_2} - 1} + \ldots +$$
$$ (1-\frac{1}{\theta})^{r-1} \cdot (1 - \frac{1}{\theta^{k_1}}) \cdot \sum\limits_{i_1 = k_r}^{\infty} \frac{1}{\theta^{i_1} - 1} \sum\limits_{i_2 = k_{r-1}}^{i_1-1} \frac{1}{\theta^{i_2}-1}\sum \cdots \sum\limits_{i_{r-1} = k_2}^{i_{r-2}-1} \frac{1}{\theta^{i_{r-1}} - 1}.$$

When $0 < \theta < 1$ and $1 \le r \le s$, we define $z_r = N-k_r$. Since $z_1 = N-k_1 \not \to \infty$,

$$\lim\limits_{N \to \infty} \frac{T_{\le r-1}(N, k_1, \ldots, k_s)}{(P_N)!} = \lim\limits_{N \to \infty} \frac{T_{\le r-1}(N, k_1, \ldots, k_r)}{(P_N)!} = \theta^{z_r} + $$
$$\theta^{z_{r-1}} \cdot (\frac{1}{\theta} - 1) \cdot \lim\limits_{N \to \infty}  \sum\limits_{i = k_r}^{N-1} 1 + \theta^{z_{r-2}} \cdot (\frac{1}{\theta} - 1)^2 \cdot \lim\limits_{N \to \infty}  \sum\limits_{i_1 = k_r}^{N-1} \sum\limits_{i_2 = k_{r-1}}^{i_1-1} 1 + \ldots +$$
$$\theta^{z_1} \cdot (\frac{1}{\theta} - 1)^{r-1} \cdot \lim\limits_{N \to \infty}  \sum\limits_{i_1 = k_r}^{N-1} \sum\limits_{i_2 = k_{r-1}}^{i_1-1} \sum \ldots \sum\limits_{i_{r-1} = k_2}^{i_{r-2}-1} 1. $$

\subsubsection{The Unconditional Expectations}

The expected number of queries made conditioned on successfully identifying the best candidate $N$ using the $(k_1, \ldots, k_s)$-strategy equals 
$$\frac{\sum\limits_{i=1}^{s} i \cdot W_i(N, k_1, \ldots, k_s)/((P_N)!)}{W(N, k_1, \ldots, k_s)/((P_N)!)} = \frac{\sum\limits_{i=1}^{s} i \cdot W_i(N, k_1, \ldots, k_s)}{W(N, k_1, \ldots, k_s)}.$$

Assume that $N \ge k_s+1$ and observe that $\frac{W(N, k_1, \ldots, k_s)}{(P_N)!}$ represents the probability of selecting the best candidate using the $(k_1, \ldots, k_s)$-strategy. Fix $k_1, \ldots, k_s$. 
Since there can be at most $s$ selections when using the $(k_1, \ldots, k_s)$-strategy, $\frac{W(N, k_1, \ldots, k_s)}{(P_N)!}$ represents the entity of interest when at most $s$ selections are made, $\frac{W(N, k_1, \ldots, k_{s-1})}{(P_N)!}$ represents the entity of interest when at most $s-1$ selections are made, $\ldots$, and $\frac{W(N, k_1)}{(P_N)!}$ represents the entity of interest when at most $1$ selection is made using the $(k_1, \ldots, k_s)$-strategy. This implies the following lemma. 
\begin{lemma}
For $N \ge k_s+1$, we have $W_1(N, k_1, \ldots, k_s) = W(N, k_1)$, and for $2 \le i \le s$, $$W_i(N, k_1, \ldots, k_s) = W(N, k_1, \ldots, k_i) - W(N, k_1, \ldots, k_{i-1}).$$
\end{lemma}

\subsubsection{Numerical Results for $s = 5$}
Note that the opportunities to query an expert are always precious. Since we also want to maintain consistency with Section~\ref{numerical}, we focus our attention on $s=5$. Both types of expectations based on the formulas from the previous subsections and for the optimal strategy described in Table~\ref{numericalresult-1} are listed in Table~\ref{table-2}. 

\subsection{The Expected Number of Candidates Interviewed}\label{expinterview}

The results for the query-based and the Dowry model differ in this setting. For the query-based model, we are informed 
when the best candidate is selected and thus stop interviewing. However, for the Dowry model, we do not have this information and will continue interviewing until the last candidate, except if we run out of selections. 

The results below pertain to the case that we are performing interviews using a $(k_1, \ldots, k_s)$-strategy, where $k_i = a_{s+1 - i}$ for each $1 \le i \le s$. As before, we examine both the unconditional expectations and the conditional expectation given that best candidate is identified.

\begin{table}
\begin{center}
\begin{tabular}{|c|c|c|c|c|c|c|c|}
\hline
$\theta$ & Unconditional  & Conditional & $\theta$ &  Unconditional  & Conditional   \\ \hline
0.01&4.95&4.95&1.01& 2.62667769 &2.6972733      \\ \hline
0.1&4.5&4.500045&1.02& 2.62685978 &2.68959942     \\ \hline
0.2&4&4.00128041&1.03& 2.59496416 & 2.66220979    \\ \hline
0.3&3.5&3.50852572&1.04& 2.79034804&2.80523622     \\ \hline
0.4&3&3.03103783&1.05& 2.64119948  & 2.68710199  \\ \hline
0.5&2.5&2.58064516&1.06&2.49825217 &2.57955508   \\ \hline
0.6& 2.7629312 &2.81568154&1.07& 3.29738824&3.17161865   \\ \hline
0.7&2.64786481&2.71674096&1.08& 3.20212123 &3.09601664  \\ \hline
0.8&2.6662117&2.73919532&1.09& 3.12062719&3.02727948    \\ \hline
0.9&2.63578947&2.71329095&1.10& 3.04382379& 2.96407481  \\ \hline
0.91&2.61849239&2.69864672&1.2 & 2.54054933  &2.52012037   \\ \hline
0.92&2.63266523&2.70893875&1.3 & 2.26194433 &  2.25566216  \\ \hline
0.93&2.63017179&2.70570293&1.4 & 2.07936215& 2.07716359   \\ \hline
0.94&2.63844464&2.71288008&1.5 & 1.94829005& 1.94744469   \\ \hline
0.95&2.62322511&2.7016097&1.6 & 1.84863362& 1.84828375   \\ \hline
0.96&2.62165305&2.69958325&1.7 &1.76979138 & 1.76963762   \\ \hline
0.97&2.6297567&2.70548721&1.8 & 1.70556686& 1.70549578   \\ \hline
0.98&2.62015186&2.69950777&1.9 & 1.65206303& 1.65202872  \\ \hline
0.99&2.61956681&2.69846554&2 & 1.60669004& 1.60667284   \\ \hline
&&&3 & 1.36430699& 1.36430692   \\ \hline
\textbf{1}&\textbf{2.61986256}&\textbf{2.69822343}&4 &1.26329306 & 1.26329306   \\ \hline
&&&5 & 1.20693541& 1.20693541   \\ \hline
\end{tabular}
\end{center}
\caption{Unconditional and conditional expected number of selections used in the optimal strategies with $\theta > 0$ and $s = 5$.}\label{table-2}
\vspace{-3mm}
\end{table}

\subsubsection{The Query-based Model}

\textbf{Case 1: Unconditional expectations.} 
Define 
$$Y'(m) = \sum\limits_{\substack{\pi \in S_N \text{ s.t. we terminate at } m \\ \text{ using the $(k_1, \ldots, k_s)$-strategy}}} \theta^{c(\pi)}.$$
We are interested in 
$$\sum\limits_{j = k_1+1}^{N} j \cdot \frac{Y'(j)}{(P_N)!}.$$

Let $k_1+1 \le m \le N$. There are three cases to consider for the position $m$ at which we stop interviewing the candidates in $\pi$.

\textbf{Case 1.1:} $k_i+1 \le m \le k_{i+1}$, where $1 \le i \le s-1$. The position $m$ must contain the value $N$ and at most $i-1$ selections can be made before position $m$, since if the former constraint does not hold we will continue interviewing until the $s^{\text{th}}$ selection and if the latter constraint does not hold then we cannot select the best candidate at position $m$ and thus will not stop at position $m$.

Therefore, 
$$Y'(m) = \theta^{N-m} \cdot B(m-1, N-m) \cdot T_{\le i-1}(k_1, \ldots, k_s) \cdot (P_{N-m})! = \frac{(P_{N-1})!}{(P_{m-1})!} \cdot \theta^{N-m} \cdot T_{\le i-1}(m-1,k_1, \ldots, k_s).$$

\textbf{Case 1.2:} $k_s+1 \le m \le N-1$. Then interviews terminate at $m$ if either $m$ is a left-to-right maxima and all $s-1$ experts were queried before position $m$ (with a final selection left for position $m$); or position $m$ has value $N$ and at most $s-2$ experts were queried before position $m$. To see this, if we stop at position $m$ then position $m$ must be a left-to-right maxima and either no selection is left after the final selection at position $m$ (i.e., we cannot continue interviewing) or position $m$ has the value $N$ and there is at least one query left before interviewing position $m$ (thus we query the expert and get the answer that we found the best candidate and stop at position $m$).

In the former case, the first $m$ positions are arbitrary elements from $\{1, \ldots, N\}$; the $m^{\text{th}}$ position has the largest value among the first $m$ positions; and exactly $s-1$ selections were used for positions in $[1, m-1]$. Now, $T_{s-1}(m-1, k_1, \ldots, k_s)$ counts the inversions within the first $m$ positions, $(P_{N-m})!$ counts the inversions within positions $[m+1, N]$, while $B(m, N-m)$ counts the inversions between the two sets. Moreover, $\theta^{N-m} \cdot T_{\le s-2}(m-1, k_1, \ldots, k_s) \cdot B(m-1, N-m) \cdot (P_{N-m})!$ accounts for the case when position $m$ has value $N$ and at most $s-2$ queries are made before position $m$.

Therefore, 
$$Y'(m) = B(m, N-m) \cdot T_{s-1}(m-1, k_1, \ldots, k_s) \cdot (P_{N-m})! + \theta^{N-m} \cdot T_{\le s-2}(m-1, k_1, \ldots, k_s) \cdot B(m-1, N-m) \cdot (P_{N-m})!$$
$$
= T_{s-1}(m-1, k_1, \ldots, k_s) \cdot \frac{(P_N)!}{(P_m)!} + \theta^{N-m} \cdot T_{\le s-2}(m-1, k_1, \ldots, k_s) \cdot \frac{(P_{N-1})!}{(P_{m-1})!}.$$

\textbf{Case 1.3:} $m = N$. Since the interviewing process must stop at $N$ we have 
$$Y'(N) = (P_N)! - \sum\limits_{j = k_1+1}^{N-1} Y'(j).$$

\textbf{Case 2: Conditional Expectations.} Define $$Y''(m) = \sum\limits_{\substack{\pi \in S_N \text{ is $(k_1, \ldots, k_s)$-winnable and} \\ \text{the search terminates at } m}} \theta^{c(\pi)}.$$
We are interested in 
$$\frac{\sum\limits_{j = k_1+1}^{N} j \cdot Y''(j)/(P_N)!}{W(N, k_1, \ldots, k_s)/(P_N)!} = \frac{\sum\limits_{j = k_1+1}^{N} j \cdot Y''(j)}{W(N, k_1, \ldots, k_s)}.$$

There are two cases to consider.

\textbf{Case 2.1:} $k_i+1 \le m \le k_{i+1}$, where $1 \le i \le s-1$. Then the value $N$ had to be at position $m$ and at most $i-1$ selections were made before position $m$, since position $m$ must contain the value $N$ in order for the process to terminate successfully, and since if $i$ selections were made before position $m$, one could not have choose the value $N$ at position $m$. Therefore, for this case 
$$Y''(m) = \theta^{N-m} \cdot T_{\le i-1}(m-1, k_1, \ldots, k_s) \cdot B(m-1, N-m) \cdot (P_{N-m})! = \frac{(P_{N-1})!}{(P_{m-1})!} \cdot \theta^{N-m} \cdot T_{\le i-1}(m-1, k_1, \ldots, k_s).$$

\textbf{Case 2.2:} $m \ge k_s+1$. Again the value $N$ had to be at position $m$ and at most $s-1$ selections had to be made before position $m$. Therefore, for this case 
$$Y''(m) = \theta^{N-m} \cdot T_{\le s-1}(m-1, k_1, \ldots, k_s) \cdot B(m-1, N-m) \cdot (P_{N-m})!  = \frac{(P_{N-1})!}{(P_{m-1})!} \cdot \theta^{N-m} \cdot T_{\le s-1}(m-1, k_1, \ldots, k_s).$$

\subsubsection{The Expectations for the Dowry Model}
In the Dowry model, the aim is to use all $s$ selections, since we do not have the information whether each of our selection is the best or not. Hence, $m \ge k_s + 1$.

\textbf{Case 1: Unconditional Expectations.} Define 
$$Z'(m) = \sum\limits_{\substack{\pi \in S_N \text{ terminate at } m \\ \text{ using the $(k_1, \ldots, k_s)$-strategy}}} \theta^{c(\pi)}.$$
We are interested in
$$\sum\limits_{j = k_s+1}^{N} j \cdot \frac{Z'(j)}{(P_N)!}.$$

\textbf{Case 1.1:} $k_s + 1 \le m \le N-1$. We only care weather the $s^{\text{th}}$ selections is made at position $m$. Thus, for each subset $S$ of $m$ values of $\{1, \ldots, N\}$, since the largest value of $S$ must be placed at position $m$ and exactly $s-1$ selections had to be made before position $m$, we have that $T_{s-1}(m-1, k_1, \ldots, k_s)$ counts the number of inversions within the first $m$ positions for each fixed $S$; $B(m, N-m)$ counts the number of inversions between the two sets in the partition of $N$, $|\Pi_1| = m$ and $|\Pi_2| = N-m$, while $(P_{N-m})!$ counts the number of inversions within positions $[m+1, N]$. Therefore, in this case, 
$$Z'(m) = T_{s-1}(m-1, k_1, \ldots, k_s) \cdot B(m, N-m) \cdot (P_{N-m})! = T_{s-1}(m-1, k_1, \ldots, k_s) \cdot \frac{(P_N)!}{(P_m)!}.$$

\textbf{Case 1.2:} $m = N$. All cases not covered by Case 1.1 include terminating at the last position and thus in this case $$Z'(m) = (P_N)! - \sum\limits_{j = k_s+1}^{N-1} Z'(j).$$

\textbf{Case 2: Conditional Expectations.} Define 
$$Z''(m) = \sum\limits_{\substack{\pi \in S_N \text{ is $(k_1, \ldots, k_s)$-winnable} \\ \text{and terminates at } m}} \theta^{c(\pi)}.$$
The entity of interest is 
$$\frac{\sum\limits_{j = k_s+1}^{N} j \cdot Z''(j)/(P_N)!}{W(N, k_1, \ldots, k_s)/(P_N)!} = \frac{\sum\limits_{j = k_s+1}^{N} j \cdot Z''(j)}{W(N, k_1, \ldots, k_s)}.$$

\textbf{Case 2.1:} Terminating at a position $k_s + 1 \le j \le N-1$ and identifying the optimal candidate. Then, the $s^{\text{th}}$ selection is made at position $m$ and position $m$ has value $N$. Therefore, in this case, 
$$Z''(j) = \theta^{N-m} \cdot T_{s-1}(m-1, k_1, \ldots, k_s) \cdot B(m-1, N-m) \cdot (P_{N-m})! = \theta^{N-m} \cdot T_{s-1}(m-1, k_1, \ldots, k_s) \cdot \frac{(P_{N-1})!}{(P_{m-1})!}.$$

\textbf{Case 2.2:} Terminating at position $N$ and identifying the optimal candidate. It is possible that not all of the $s-1$ selections were used before and the best candidate was selected at position $N$; it is also possible that the best candidate was picked at some position in $[k_{s+1}, N-1]$ but not all selections were used before the position $N$ and thus the search continued until after position $N$ (those two possibilities do not account for all possible cases). In this case,
$$\frac{Z''(N)}{(P_N)!} = \frac{W(N, k_1, \ldots, k_s)}{(P_N)!} - \sum\limits_{j = k_s + 1}^{N-1} \frac{Z''(j)}{(P_N)!}.$$

\subsubsection{Numerical Results for $s = 5$}
We again focus our attention on $s=5$. Both types of expectations based on the formulas from the previous subsections and for the optimal strategy described in Table~\ref{numericalresult-1} are listed in Table~\ref{table-2}.

\vspace{5mm}
\noindent\textbf{Acknowledgment.} The work was supported in part by the NSF grants NSF CCF 15-26875, CIF\,1513373, through Rutgers University, and The Center for Science of Information at Purdue University, under contract number 239 SBC PURDUE 4101-38050. The work was done while X. Liu was with the University of Illinois, Urbana-Champaign.

\section{Appendix}\label{appendix}

We make use of the following two results.

\begin{theorem}[Jones~\cite{jones2020weighted}, Theorem 6.5]\label{theta<1}
For $\theta < 1$, the optimal asymptotic selection strategy for the secretary problem is to reject $N-j(\theta) \not \to \infty$ (not depending on $N$) candidates and then accept the next left-to-right maxima thereafter.
\end{theorem}

\begin{theorem}[Jones~\cite{jones2020weighted}, Corollary 6.6]\label{theta>1}
For $\theta > 1$, the optimal asymptotic selection strategy for the secretary problem is to reject $j(\theta) \not \to \infty$ (not depending on $N$) candidates and then accept the next left-to-right maxima thereafter.
\end{theorem}

\textbf{Proof of Theorem~\ref{asm-s}.} The proof follows by induction.

First, note that $N-k_1 \not \to \infty$ implies $N-k_2, \ldots, N-k_s \not \to \infty$. The base case for one threshold hold by Theorem~\ref{theta<1}. We assume the argument works for $ \le s-1$ thresholds and prove the result for $s$-thresholds with $s \ge 2$. 

By Lemma~\ref{relations-s} and the fact $T_{\le s-1}(k_s, k_1, \ldots, k_{s-1}, k_s) = (P_{k_s})!$, we have 
$$\frac{W(N, k_1, \ldots, k_s)}{(P_N)!} = \frac{1}{P_N} \cdot  \Big( \frac{\theta^{N-k_s} \cdot P_{k_s} \cdot W(k_s, k_1, \ldots, k_{s-1})}{(P_{k_s})!} + \frac{T_{\le s-1}(N-1, k_1, \ldots, k_s)}{(P_{N-1})!}+$$
\begin{equation}
\theta \cdot \frac{T_{\le s-1}(N-2, k_1, \ldots, k_s)}{(P_{N-2})!} + \ldots + \theta^{N-k_s-2} \cdot \frac{T_{\le s-1}(k_s+1, k_1, \ldots, k_s)}{(P_{k_s + 1})!} + \theta^{N-k_s-1} \cdot 1  \Big).
\end{equation}

\textbf{Step 1:} Suppose that $N-k_s \not \to \infty$ does not hold, i.e., that $N-k_s \to \infty$. 
Then, since the probability
$$\frac{W(k_s, k_1, \ldots, k_{s-1})}{(P_{k_s})!} \le 1, \quad \text{ we have }$$ 
$$\frac{\theta^{N-k_s} \cdot P_{k_s} \cdot W(k_s, k_1, \ldots, k_{s-1})}{(P_{k_s})!} \to \frac{\theta^{N-k_s}}{1-\theta} \cdot \frac{W(k_s, k_1, \ldots, k_{s-1})}{(P_{k_s})!} \to 0.$$
Moreover, by Lemma~\ref{formula-t}, for $0 \le j \le N-k_s-2$,
\begin{align}\label{formula-t2}
\theta^j \cdot \frac{T_{\le s-1}(N-j-1, k_1, \ldots, k_s)}{(P_{N-j-1})!} = \frac{1}{P_{N-j-1}} \cdot  \Big( \theta^{N-k_{s}-1} \cdot P_{k_s} + \theta^{N-k_{s-1}-2} \cdot P_{k_{s-1}} \cdot \sum\limits_{i = k_s}^{N-j-2} \frac{1}{P_i}+
\end{align}
$$\theta^{N-k_{s-2}-3} \cdot P_{k_{s-2}} \cdot \sum\limits_{i_1 = k_s}^{N-j-2} \frac{1}{P_{i_1}} \sum\limits_{i_2 = k_{s-1}}^{i_1-1} \frac{1}{P_{i_2}}  + \ldots + \theta^{N-k_1-s} \cdot P_{k_1} \cdot \sum\limits_{i_1 = k_s}^{N-j-2} \frac{1}{P_{i_1}} \cdot \sum\limits_{i_2 = k_{s-1}}^{i_1-1} \frac{1}{P_{i_2}} \sum \cdots \sum\limits_{i_{s-1} = k_2}^{i_{s-2}-1} \frac{1}{P_{i_{s-1}}}  \Big).$$
For each term inside the bracket of~\eqref{formula-t2}, since $N-k_i \to \infty$, $\theta^{N-k_i} \to 0$ exponentially, and the sum part (without the multiplier) of each term approaches infinity as a polynomial function in $N-k_{i+1} \le N-k_{i}$. The latter claim holds since $\frac{1}{P_i} = \frac{1-\theta}{1-\theta^i} \le 1$ when $i \ge 1$ and the smallest value of the subscript equals $k_2 \ge 1$. 

Since $N-j-1 \ge k_s+1$, we have
$$\theta^j \cdot \frac{T_{\le s-1}(N-j-1, k_1, \ldots, k_s)}{(P_{N-j-1})!} \le (\theta^{N-k_s-1} + \theta^{N-k_{s-1}-2} \cdot (N-k_{s-1}) + \theta^{N-k_{s-2}-3} \cdot (N-k_{s-2})^2 + \ldots$$
$$ +  \theta^{N-k_{1}-s} \cdot (N-k_{1})^{s-1}) \le s \cdot \max\limits_{q=1}^{s}\{\theta^{N-k_{q}-(s+1-q)} \cdot (N-k_{q})^{s-q}\} \to 0.$$
Above, the convergence rate $\to 0$ is exponential. Thus, 
$$
\frac{T_{\le s-1}(N-1, k_1, \ldots, k_s)}{(P_{N-1})!}+
\theta \cdot \frac{T_{\le s-1}(N-2, k_1, \ldots, k_s)}{(P_{N-2})!} + \ldots + \theta^{N-k_s-2} \cdot \frac{T_{\le s-1}(k_s+1, k_1, \ldots, k_s)}{(P_{k_s + 1})!} + \theta^{N-k_s-1}$$
$$
\le (N-k_s) \cdot s \cdot \max\limits_{q=1}^{s}\{\theta^{N-k_{q}-(s+1-q)} \cdot (N-k_{q})^{s-q}\} 
\le s \cdot \max\limits_{q=1}^{s}\{\theta^{N-k_{q}-(s+1-q)} \cdot (N-k_{q})^{s+1-q}\} \to 0. 
$$

\textbf{Step 2:} Suppose that $N-k_1 \not \to \infty$ does not hold, i.e., that $N-k_1 \to \infty$.  

From Step 1, we know that for an optimal strategy, $N-k_{s} \not \to \infty$ has to hold. We also know that $k_s-k_1 \to \infty$. By the induction hypothesis, 
\begin{equation}\label{reason-1}
\frac{W(k_s, k_1, \ldots, k_{s-1})}{(P_{k_s})!} \text{ is maximized when } k_s-k_1 \not \to \infty (\text{ and is not maximized when } k_s-k_1 \to \infty).
\end{equation}
Moreover, for each $0 \le j \le N-k_s-1$, by Lemma~\ref{formula-t} and~\eqref{formula-t2},
\begin{equation}\label{reason-2}
\theta^j \cdot \frac{T_{\le s-1}(N-j-1, k_1, \ldots, k_s)}{(P_{N-j-1})!} \text{ is maximized when } N-k_s, \ldots, N-k_1 \not \to \infty.
\end{equation}
To see why this is the case, say $N-k_i \to \infty$ and $N-k_{i+1} \not \to \infty$. Then the term involving $\theta^{N-k_q-(s+1-q)}$ approaches zero for $1 \le q \le i$ (when $k_1 = 0$, the last term is zero by default).

Therefore, by~\eqref{reason-1} and~\eqref{reason-2}, $$\frac{W(N, k_1, \ldots, k_s)}{(P_N)!} \text{ is maximized only when }N-k_1 \not \to \infty,$$
which also implies $N-k_1 \not\to \infty, \ldots, N-k_s \not\to \infty$. \hfill\qed

\textbf{Proof of Theorem~\ref{asm-l}:} The proof proceeds by induction. For consistency, let $k_0 = -1$ and $k_{s+1} = N$.

Note that $k_s \not \to \infty$ implies $k_{s-1}, \ldots, k_1 \not \to \infty$. By Theorem~\ref{theta>1}, the argument works for one threshold. Suppose now that the argument works for at most $s-1$ thresholds. We prove the claimed result for $s$ thresholds. To this end, we show that the choice $k_i \not \to \infty$ and $k_{i+1} \to \infty$ is always at least as good as the choice $k_{i-1} \not \to \infty$ and $k_i \to \infty$, where $1 \le i \le s-1$. 

\begin{claim}\label{converge}
For the case $k_i \to \infty$ and $k_{i-1} \not \to \infty$, where $1 \le i \le s$ we have
$$\frac{W(N, k_1, \ldots, k_s)}{(P_N)!} \to \frac{ W(k_i, k_1, \ldots, k_{i-1})}{(P_{k_i})!}.$$
\end{claim}

\begin{proof}
By Lemma~\ref{relations-s} and $T_{\le s-1}(k_s, k_1, \ldots, k_{s-1}, k_s) = (P_{k_s})!$, 
$$\frac{W(N, k_1, \ldots, k_s)}{(P_N)!} = \frac{1}{P_N} \cdot  \Big( \frac{\theta^{N-k_s} \cdot P_{k_s} \cdot W(k_s, k_1, \ldots, k_{s-1})}{(P_{k_s})!} +  \frac{T_{\le s-1}(N-1, k_1, \ldots, k_s)}{(P_{N-1})!} +$$
$$
\theta \cdot \frac{T_{\le s-1}(N-2, k_1, \ldots, k_s)}{(P_{N-2})!}  + \ldots + \theta^{N-k_s-2} \cdot \frac{T_{\le s-1}(k_s+1, k_1, \ldots, k_s)}{(P_{k_s + 1})!} + \theta^{N-k_s-1}  \Big).
$$
$$
= \quad \cdots =
$$
$$
\frac{1}{P_N} \cdot  \Big( \frac{\theta^{N-k_i} \cdot P_{k_i} \cdot W(k_i, k_1, \ldots, k_{i-1})}{(P_{k_i})!} +  
$$
$$
\frac{T_{\le s-1}(N-1, k_1, \ldots, k_s)}{(P_{N-1})!} +\theta \cdot \frac{T_{\le s-1}(N-2, k_1, \ldots, k_s)}{(P_{N-2})}  + \ldots + \theta^{N-k_s-2} \cdot \frac{T_{\le s-1}(k_s+1, k_1, \ldots, k_s)}{(P_{k_s + 1})!} + \theta^{N-k_s-1}
$$
$$
+ \quad \cdots \quad +
$$
\begin{equation}\label{prob2}
+ \theta^{N-k_{i+1}} \cdot \frac{T_{\le i}(k_{i+1}-1, k_1, \ldots, k_{i})}{(P_{k_{i+1} - 1})!} + \ldots + \theta^{N-k_i-2} \cdot \frac{T_i(k_i + 1, k_1, \ldots, k_i)}{(P_{k_i+1})!} + \theta^{N-k_i-1}
 \Big).
\end{equation}

Since $T_j(m, k_1, \ldots, k_{j+1}) \le (P_m)!$,
$$
\frac{T_{\le s-1}(N-1, k_1, \ldots, k_s)}{(P_{N-1})!} +\theta \cdot \frac{T_{\le s-1}(N-2, k_1, \ldots, k_s)}{(P_{N-2})!}  + \ldots + \theta^{N-k_s-2} \cdot \frac{T_{\le s-1}(k_s+1, k_1, \ldots, k_s)}{(P_{k_s + 1})!} + \theta^{N-k_s-1}
$$
$$
+ \quad \cdots  \quad +
$$
$$
+ \theta^{N-k_{i+1}} \cdot \frac{T_{\le i}(k_{i+1}-1, k_1, \ldots, k_{i})}{(P_{k_{i+1} - 1})!} + \ldots + \theta^{N-k_i-2} \cdot \frac{T_i(k_i + 1, k_1, \ldots, k_i)}{(P_{k_i+1})!} + \theta^{N-k_i-1}
$$
$$
\le 1 + \theta + \ldots + \theta^{N-k_i-1}= \frac{\theta^{N-k_i} - 1}{\theta - 1},
$$
and thus $~\eqref{prob2}$ converges to
$$
\frac{1}{1-1/(\theta^N)} \cdot  \Big(  (1-1/\theta^{k_i}) \cdot \frac{ W(k_i, k_1, \ldots, k_{i-1})}{(P_{k_i})!} +  (1/\theta^{k_i} - 1/(\theta^N))  \Big) \to \frac{ W(k_i, k_1, \ldots, k_{i-1})}{(P_{k_i})!} \text{ as } k_i \to \infty. \qedhere
$$
\end{proof}

We compare the class of $(k_1, \ldots, k_s)$-strategies for which $k_{i-1} \not \to \infty$ (Case 1) and $k_i \to \infty$, with the class of $(k_1, \ldots, k_s)$-strategies for which $k_{i} \not \to \infty$ and $k_{i+1} \to \infty$ (Case 2). 

\begin{claim}\label{better}
For every strategy $(k_1', \ldots, k_{s}')$ covered under Case 1 there is a strategy covered under Case 2 which performs better.
\end{claim}
\begin{proof}
Let $(k_1', \ldots, k_{s}')$ be a strategy with $k_{i-1}' \not \to \infty$ and $k_i' \to \infty$ (Case 1) and let $(k_1'', \ldots, k_{s}'')$ be a strategy such that $k_j'' = k_j'$ for $j \in \{1, \ldots, s\} - \{i\}$, $k_i'' \not \to \infty$. Note that $k_i'' \not\to \infty, k_{i+1}'' \to \infty$ and thus the $(k_1'', \ldots, k_s'')$-strategy is also covered under Case 2.

By Claim~\ref{converge}, the probability of success for the $(k_1', \ldots, k_s')$-strategy is
$$
\lim\limits_{k_i \to \infty} \frac{ W(k_i', k_1', \ldots, k_{i-1}')}{(P_{k_i'})!} < \lim\limits_{k_{i+1}'' \to \infty} \frac{ W(k_{i+1}'', k_1'', \ldots, k_{i-1}'', k_{i}'')}{(P_{k_{i+1}''})!},
$$
which is the winning probability for the $(k_1'', \ldots, k_s'')$-strategy since the latter case has one more selection than the former case and the value $k_i''$ can be placed anywhere as long as $k_i'' \ge k_{i-1}''$ and $k_i'' \not \to \infty$. \qedhere
\end{proof}

By Claim~\ref{better}, the proof follows. \hfill\qed

\end{document}